\documentclass[11pt]{article}

\usepackage{color}

\usepackage{setspace}
\usepackage{natbib}
\usepackage{xparse}
\usepackage{amsmath}
\usepackage{color}
\usepackage{amssymb}
\usepackage{amsfonts}
\usepackage{multirow}
\usepackage{amsthm}
\usepackage{qtree}
\usepackage{booktabs}
\usepackage{changepage}
\usepackage{verbatim}
\usepackage{subcaption}
\usepackage{mathrsfs}
\usepackage[shortlabels]{enumitem}
\usepackage{tikz}
\usetikzlibrary{shapes,arrows,positioning}
\usepackage{dsfont}
\usepackage{graphics}
\usepackage{boldline}
\usepackage{graphicx}
\usepackage[margin=1in]{geometry}
\usepackage{natbib}

\newtheorem{thm}{Theorem}
\newtheorem{lemma}{Lemma}

\newtheorem{remark}{Remark}

\newtheorem*{example*}{Example \DO}

\newcommand{\vW}{{\bf W}}
\newcommand{\vone}{\boldsymbol{1}}

\newcommand{\vm}{{\bf m}}
\newcommand{\vt}{{\bf t}}

\newcommand{\vx}{{\bf x}}

\newcommand{\vX}{{\bf X}}

\newcommand{\vv}{{\bf v}}
\newcommand{\vA}{{\bf A}}
\newcommand{\vK}{{\bf K}}

\newcommand{\vY}{{\bf Y}}
\newcommand{\vV}{{\bf V}}

\newcommand{\vepsilon}{\mbox{\boldmath $\epsilon$}}
\newcommand{\vtheta}{\mbox{\boldmath $\theta$}}
\newcommand{\vnu}{\mbox{\boldmath $\nu$}}

\newcommand{\valpha}{\mbox{\boldmath $\alpha$}}

\newcommand{\pkonv}{\stackrel{p}{\rightarrow}}
\newcommand{\dkonv}{\stackrel{d}{\rightarrow}}

\newcommand{\bqa}{\begin{eqnarray*}}
\newcommand{\eqa}{\end{eqnarray*}}
\newcommand{\bqan}{\begin{eqnarray}}
\newcommand{\eqan}{\end{eqnarray}}
\newcommand{\bit}{\begin{itemize}}
\newcommand{\eit}{\end{itemize}}
\newcommand{\ben}{\begin{enumerate}}
\newcommand{\een}{\end{enumerate}}
\newcommand{\beq}{\begin{equation}}
\newcommand{\eeq}{\end{equation}}
\newcommand{\bdes}{\begin{description}}
\newcommand{\edes}{\end{description}}

\usepackage{graphicx} 
\newtheorem{theorem}{Theorem}

\title{Nonparametric Testing and Variable Selection for ARCH-$m(X)$ Model}
\author{Adriano Zambom and Qing Wang}
\date{}

\begin{document}

\maketitle

\begin{abstract}
    We introduce the ARCH-$m(X)$ model, a semiparametric extension of the ARCH-$X$
framework in which the effect of a multivariate exogenous covariate vector
$\mathbf{X} \in \mathbb{R}^d$ on the conditional variance is modeled through
an unknown nonparametric function $m(\cdot)$, accommodating complex nonlinear
relationships between external predictors and financial volatility. Within this model, we develop a novel hypothesis test for the significance of covariates constructed with an artificial one-way ANOVA. Under some regularity
conditions, the test statistic is shown to converge in distribution to the standard
Normal. 
Another key contribution of this paper is the construction of a variable selection procedure based on the Benjamini-Yekutieli false discovery rate correction applied to covariate-level $p$-values. We show that the resulting index set coincides with the true set of relevant covariates with probability tending to one as $n \to \infty$.
Extensive simulations confirm that the proposed
methods outperform existing competitors, and an empirical application to
S\&P~500 return volatility illustrates the practical utility of the proposed variable selection framework.

\end{abstract}

\section{Introduction}

The Autoregressive Conditional Heteroscedasticity (ARCH) model introduced by \cite{Engle1982} and its generalization to the 
Generalized Autoregressive Conditional Heteroskedasticity (GARCH) framework \citep{Bollerslev1986} have become the foundational tools for modeling time-varying volatility in financial time series. Their theoretical development and empirical relevance continue to attract substantial research attention; see, for example, \cite{Andersen2003, Hansen2005, BollerslevEtAl2016, FrancqZakoian2019, FrancqZakoian2022} and references therein.

A growing body of empirical evidence suggests that financial volatility is often systematically driven by exogenous factors beyond the endogenous dynamics captured by standard ARCH and GARCH specifications \citep{EnglePatton2001}. Variables such as macroeconomic indicators and commodity prices have been shown to significantly influence volatility in equity markets and currency exchange rates \citep{sidorov2014garch, sadik2018news, chua2019information}. This has motivated the ARCH-$X$ and GARCH-$X$ class of models, in which exogenous covariates enter the conditional variance equation directly. The theoretical foundations for inference in these models have attracted recent interest \citep{FrancqZakoian2007, HanPark2012, HanKristensen2014, Han2015, FrancqThieu2019, Sucarrat2021, WuEtAl2024}, 
and an expanding literature continues to identify new applications for these models
\citep{FernandesEtAl2014, ConradEtAl2018, WangEtAl2020, AudrinoEtAl2020, YamakaEtAl2023}.

A fundamental limitation of existing ARCH-$X$ and GARCH-$X$ formulations is that the relationship between exogenous covariates and the conditional variance is assumed to belong to a prescribed parametric family, typically linear. This restriction may be overly conservative: financial volatility can exhibit complex, nonlinear responses to external drivers, and functional form misspecification leads to biased estimation and invalid inference. Moreover, the problem of variable selection among exogenous covariates remains unresolved in volatility models with nonparametric components. Although model selection for time series has been addressed in several contributions \citep{BardEtAl2020, DiopKengne2022, Kengne2021, BardetEtAl2023}, existing procedures target lag-order selection in Autoregressive Moving Average (ARMA) and GARCH-type models and are not directly designed for variable selection among covariates. Penalized likelihood approaches analogous to BIC can in principle be applied, but exhaustive search over covariate subsets becomes computationally challenging even for a moderate number of candidates.

This paper addresses both limitations simultaneously. First, we introduce the ARCH-$m(X)$ model, a semiparametric ARCH-$X$ framework in which the effect of a multivariate exogenous covariate vector $\mathbf{X} \in \mathbb{R}^d$ on the conditional variance is modeled through an unknown nonparametric function $m(\cdot)$. 
A central contribution of this paper is a novel hypothesis test for the relevance of individual exogenous covariates within the ARCH-$m(X)$ framework. The test is constructed by utilizing an artificial one-way ANOVA in which the number of factor levels grows with sample size, and we show that the resulting test statistic converges in distribution to a standard Normal under regularity conditions. 

Building on the novel test statistic, we develop a variable selection procedure for the ARCH-$m$(X) model by applying the Benjamini--Yekuteli \citep{BenjaminiYekutieli2001} false discovery rate (FDR) correction to the covariate-level $p$-values. This exploits the classical connection between hypothesis testing and variable selection: excluding covariate $j$ is equivalent to not rejecting $H_0^j : m_j = 0$, where $m_j$ denotes the marginal contribution of covariate $X_j$ \citep{AbramovichEtAl2006, KongXia2007, LiLiang2008, WangXia2009, HuangEtAl2010, StorlieEtAl2011, BuneaEtAl2006}. We formally establish the consistency of the selection procedure, showing that the selected index set coincides with the true relevant covariate set with probability tending to one as $n \to \infty$. Extensive simulations demonstrate favorable performance relative to competing methods, and an empirical application to S\&P~500 return volatility with commodity prices as candidate covariates illustrates the practical utility of the proposed variable selection framework.

The remainder of the paper is organized as follows. Section~\ref{sec:method} introduces the ARCH-$m(X)$ model and develops the estimation and the hypothesis test procedure. The limiting distribution of the proposed test statistic is derived in Section \ref{sec.Test_asymptotics}. Section~\ref{sec.variable_selection} introduces the variable selection procedure and establishes consistency based on FDR corrections. Sections~\ref{sec.simulations} and~\ref{sec.real_data} present the simulation studies and an empirical application, respectively.

\section{Nonparametric Hypothesis Test}\label{sec:method}

\subsection{Methodology}
Consider an Autoregressive Conditionally Heteroscedastic model of order $p$ and  $d$ exogenous covariates $X_{1}, \ldots, X_d$, which is often referred to as an ARCH($p$)-$m(\mathbf{X})$ model. Given observations $\epsilon_t\ (1\leq t\leq n)$, the model is specified in the following form:
\begin{equation}
    \begin{cases}
  \epsilon_t = z_t\sigma_t,\\
  \sigma_t^2= \sum_{j=1}^p \alpha_j\epsilon_{t-j}^2 +m(\mathbf{X}_{t-1}),
    \end{cases}
    \label{eq:ARCHmodel}
\end{equation}
where $m$ denotes a nonparametric function that includes the exogenous covariates $\mathbf{X}_t = (X_{1t},\ldots,X_{dt})^T$, $z_t$'s are i.i.d. random variables with zero mean and unit variance and are independent of $\mathbf{X}_t$. In addition, $\alpha_j\geq 0$ for $j=1,\ldots, p$, and $m(\mathbf{X}_t)>0$. Furthermore, assume that $E(\epsilon_t|\mathcal{I}_{t-1})=\sigma_t^2>0$, where $\mathcal{I}_{t-1}$ is the $\sigma$-algebra generated by $\{\epsilon_u,\mathbf{X}_u: u < t\}$, and $E(z_t^4)<\infty$. Note that this model does not contain a constant term in the volatility as it is absorbed into the function $m(\cdot)$.

By setting $v_t=\epsilon_t^2-E(\epsilon_t^2|\mathcal{I}_{t-1})=\epsilon_t^2-\sigma_t^2$, equation \eqref{eq:ARCHmodel} can be rewritten as an autoregressive model of order $p$ in the following way
\begin{equation}    
\epsilon_t^2 =  \sum_{j=1}^p\alpha_j\epsilon_{t-j}^2 + m(\mathbf{X}_{t-1}) + v_t.
    \label{eq:model-new}
\end{equation}
The errors $v_t$'s in equation \eqref{eq:model-new} are no longer i.i.d.. They are white noises with variance given by
{\small
\begin{eqnarray*}
    \tau&=&\text{Var}(z_t^2)E\left[ \sum_{j=1}^p\alpha_j\epsilon_{t-j}^2 + m(\mathbf{X}_{t-1})\right]^2\\
    &=&\text{Var}(z_t^2)\Big[E\left(\sum_{j=1}^p\alpha_j \epsilon_{t-j}^2 \right)^2+E\{m^2(\mathbf{X}_{t-1})\}\\
    && + 2E\{m^2(\mathbf{X}_{t-1})\}\left(\sum_{j=1}^p \alpha_j E(\epsilon_{t-j}^2) \right)\Big]
\end{eqnarray*}}
Additionally, since $E(v_t|\mathcal{T}_{t-1})=0$ with $\mathcal{T}_{t-1}$ being a $\sigma$-algebra generated by $(v_{t-1},v_{t-2},\ldots)$, $v_t$ is a martingale difference. Moreover, given the assumption on $z_t$ specified in model \eqref{eq:ARCHmodel}, $E(v_t^2)<\infty$.

Assume that each $X_\ell$, among the $d$ exogenous covariates $\{X_1,\ldots,X_d\}$, is distributed over a compact interval, say $[a_\ell, b_\ell]$. Without loss of generality, assume these compact intervals are all $[0,1]$. Let $\mathcal{G}_\ell$ be the space of polynomial splines of order $m_\ell$. Denote $B_\ell(x_\ell)=\{B_{\ell,j_\ell}(x_\ell): 1-m_\ell\leq j_\ell<N_\ell\}^T$ as a basis system of space $\mathcal{G}_\ell$ for the $\ell$th coordinate $X_\ell$  ($1\leq \ell\leq d)$, where $N_\ell$ is the number of internal knots and $m_\ell$ is the spline order. Thus, the number of basis functions in $B_\ell$ is $K_\ell=N_\ell+m_\ell$. Define the space of tensor-product polynomial splines by $\mathcal{G}=\otimes_{\ell=1}^d\mathcal{G}_\ell$ which is a linear space of dimension $K=\prod_{\ell=1}^d K_\ell$. Then, we assume that the function $m(\cdot)$ in equation \eqref{eq:ARCHmodel} belongs to the space $\mathcal{H}_{K}$ with interior knot sequence $\vt := \vt_\ell = t_{0,\ell} = 0 < t_{1,\ell} < \ldots < t_{N_\ell + m_\ell, \ell} $ that is spanned by a B-spline basis system defined as a tensor product of $B_\ell$ for $1\leq \ell\leq d$. That is,
\bqan \label{m_param}
m(\vx) = \vtheta^T \mathcal{B}(\vx),
\eqan
where 
\bqa
\mathcal{B}(\vx) = \left[\left\{\mathcal{B}_{\ell_1, \ldots, \ell_d}(\vx)\right\}_{\ell_1 = 1-m_1, \ldots, \ell_d = 1-m_d}^{N_1, \ldots, N_d}\right]_{K\times 1} = B_1(x_1) \otimes \cdots \otimes B_d(x_d)
\eqa
is a basis system of the space $\mathcal{G}$, and $\vtheta = (\theta_1, \ldots, \theta_{K})$ is a vector of parameters. With sufficiently large $K$, $\mathcal{H}_{K}$ approximates the infinite-dimensional Sobolov space $\mathcal{H}_2^2 = \left\{f:[0,1] \to \mathbb{R}, \int f^2 + \int (f')^2 + \int(f'')^2 < \infty\right\}$, as shown in \cite{reif1997orthogonality}.
With this formulation, the volatility in (\ref{eq:ARCHmodel}) becomes
\bqa
  \sigma_t^2= \sum_{j=1}^p \alpha_j\epsilon_{t-j}^2 + \vtheta^T \mathcal{B}(\vx).
\eqa
To estimate the parameter vector $\vnu = (\alpha_1, \ldots, \alpha_p,  \theta_1, \ldots, \theta_K)$ assume we have $n$ observations of the time series $\epsilon_t$ and covariates $\vx_t$ and, given initial values $\epsilon_{1-p}, \ldots, \epsilon_0$ and $\sigma_{1-q}, \ldots, \sigma_0$, consider the quasi-likelihood 
\bqa
L_n(\vnu) = L_n(\vnu, \epsilon_1, \ldots, \epsilon_n, \vx_1, \ldots, \vx_T) = \prod_{t=1}^n\frac{1}{\sqrt{2\pi\sigma^2_t}}\exp\left(\frac{-\epsilon_t^2}{2\sigma_t^2}\right).
\eqa
The QMLE is then given by
\bqa
\hat{\vnu} = \arg\max_{\vnu}L_n(\vnu) = \arg\min_{\vnu}\frac{1}{n}\sum_{t = 1}^n\frac{\epsilon_t^2}{\sigma_t^2} + \ln\sigma_t^2,
\eqa
which is $\sqrt{n}$ consistent (see \cite{FrancqThieu2019}).



The main objective is to test the significance of an exogenous covariate, say $X_\ell$, under this nonparametric ARCH framework. Let
$\mathbf{X}^{-\ell}=(X_1,\ldots,X_{\ell-1}, X_{\ell+1},\ldots, X_d)^T$ be a vector of $d-1$ covariates excluding $X_\ell$. To test for the significance of $X_\ell$, the null hypothesis of interest can be expressed as
\begin{equation}
H_0: m(\mathbf{X}_t) = m_\ell(\mathbf{X}^{-\ell}_t),\label{eq:H0}
\end{equation}
where $m_\ell$ denotes a nonparametric functional.
Under the null hypothesis (\ref{eq:H0}) the residuals are defined as
\bqan
    \hat{v}_t=\epsilon_{t}^2 -\sum_{j=1}^p\hat{\alpha}_j\epsilon_{t-j}^2 -\hat{m}_\ell(\mathbf{X}^{-\ell}_{t-1})
    \label{eq.v_hat}
\eqan
for $t=1,\ldots,n$. Under $H_0$, $\hat{v}_t$'s are approximately white noises, and therefore they are uncorrelated with the covariate $\mathbf{X}^{-\ell}_t$. Hence, we can construct a test statistic for \eqref{eq:H0} by gauging the strength of the association between $\hat{v}_t$ and $X_{\ell t}$.

Motivated by the work of \cite{WangEtAl2010}, we propose to build a one-way ANOVA-type of design to test for $H_0$, where the factor levels are composed in a specific way as detailed below. Without loss of generality, assume $(X_{\ell t},\hat{v}_t)$ are arranged in such a way that $X_{\ell t_1} < X_{\ell t_2}$ for any $t_1<t_2$. Let $(X_{\ell t},\hat{v}_t)$ be data from a balanced one-way ANOVA design with $\hat{v}_t$ represent the data at ``level" $X_{\ell t}$. Since ANOVA statistics require multiple observations at each factor level, we augment each level by including $\hat{v}_s$ values from the $k_n-1$ nearest covariate values of $X_{\ell t}$. Here $k_n$, a fixed constant, stands for the number of observations in each factor level. We refer to it as the window size. Specifically, we define the windows as
\[
W_t=\left\lbrace s: |\hat{F}_X(X_{\ell s})-\hat{F}_X(X_{\ell t})|\leq \frac{k_n-1}{2n}\right\rbrace,
\]
where $\hat{F}_X$ is the empirical cumulative distribution of $X_{\ell}$. At edges, i.e., for the lowest and highest $k_n$ values of the covariate, asymmetric windows can be used, which have negligible effects on the asymptotic distribution of the test statistic described next.

We propose to use the following test statistic to test for $H_0$ \eqref{eq:H0}:
\begin{equation} \label{eq:Tn}
  T_n=MST_n-MSE_n=\frac{k_n}{n-1}\sum_{t=1}^n (\hat{v}_{t\cdot}-\hat{v}_{\cdot\cdot})^2-\frac{1}{n(k_n-1)}\sum_{t=1}^n\sum_{s\in W_{t}}(\hat{v}_s-\hat{v}_{t\cdot})^2,  
\end{equation}
where $\hat{v}_{t\cdot}=(1/k_n)\sum_{s\in W_t}\hat{v}_s$ and $\hat{v}_{\cdot\cdot}=(1/nk_n)\sum_{t=1}^n\sum_{s\in W_t}\hat{v}_s$.

\subsection{Asymptotic Results} \label{sec.Test_asymptotics}

Consider the following conditions.

C1: the function $m(\cdot)$ belongs to the space $\mathcal{H}_{K}$ as defined in equation (\ref{m_param}). 

C2: Under $H_0$, $\hat{\vnu} - \vnu = O_p(n^{-1/2})$.

C3: Density $f_X(\cdot)$ of covariate $X$ is twice continuously differentiable in a compact support $\mathcal{X}$ with $f(x) > 0$ for $x \in \mathcal{X}$.

C4: $\tau < \infty$.

Conditions C1 and C2 guarantee the convergence rate of the estimators of the parameters in the ARCH model which, in turn, allows specific terms in the test statistic to shrink or stabilize asymptotically, allowing the null distribution to be Gaussian asymptotically, as established in Theorem \ref{thm.null}.
Condition C3 typically holds for probability density functions with compact support and constitutes the only restriction on the covariate $\mathbf{X}$. Next, we state the asymptotic distribution result of the test statistic $T_n$ under $H_0$ and local alternatives.

\begin{theorem} \label{thm.null}
    Assume considers C1-C4 hold . Then, under $H_0$ in \eqref{eq:H0}, as $n\to\infty$, $k_n\to \infty$, and $k_n$ is such that $k_n^{5/2}/n^{1/2}\to 0$, 
    \[
    \left(\frac{n}{k_n}\right)^{1/2}T_n\xrightarrow{D} N(0, 4\tau^2/3).
    \]
\end{theorem}

The proof of Theorem \ref{thm.null} can be found in the appendix.
 We estimate the variance $\tau$ of $v_t$ by applying \cite{Rice1984}'s
estimator to the filtered data
 \bqa
 \hat{\tau} = \frac{1}{2(n-1)}\sum_{t=2}^n (\hat{v}_t- \hat{v}_{t-1})^2.
 \eqa
Because $\hat\alpha_j$ and $\hat\vtheta$ are $\sqrt{n}$-consistent estimators, we have that $\hat{\tau}$ is a consistent estimator of $\tau$.

The asymptotic power of the proposed test statistic can be studied by considering the
probability the test rejects the null hypothesis when the alternative hypothesis approaches the null at a
certain rate, also known as Pitman alternatives.
In order to obtain an insight on the power of the proposed  statistic, we consider a sequence of local additive alternatives and general local alternatives as follows
\bqan\label{Ha}
&&H^A_a: m(\vX_t) = m_\ell(\vX_t^{-\ell}) + (nk_n)^{-1/4}\tilde{m}_\ell(X_{\ell t}) \nonumber\\
&&H^G_a: m(\vX_t) = m_\ell(\vX_t^{-\ell }) + (nk_n)^{-1/4}(\tilde{m}_\ell (X_{\ell t}) + \tilde{m}(\vX_{t})),
\eqan
where the functions $\tilde{m}_\ell$ and $\tilde{m}$ are assumed to be such that $E(\tilde{m}(\vX_t))  = E(\tilde{m}_\ell(X_{\ell t})) = 0$.

Theorem~\ref{thm.altern} establishes that under local additive and general alternatives as specified in (\ref{Ha}), the proposed test statistic asymptotically follows a Normal distribution and can detect effects at a rate of $(nk_n)^{-1/4}$. The deviation of the alternative distribution from the null is reflected in the asymptotic mean, indicating that the test's power arises from the variation in the local function $\tilde{m}_\ell(\cdot)$ (and $\tilde{m}(\cdot)$ in the general case) deviating from the null model.

\begin{theorem} \label{thm.altern} Assume conditions C1-C4 hold and let $\tilde{m}_i(\cdot)$ be Lipschitz continuous. Then, as $n\to\infty$, $k_n\to\infty$ and $k_n$ is such that $k_n^{5/2}/n^{1/2} \to 0$
\begin{enumerate}
    \item under the additive alternative $H_a^A$we have
\bqa
\left(\frac{n}{k_n}\right)^{1/2}T_n \dkonv N\left(\text{Var}(\tilde{m}_\ell(X_{\ell t})) , \frac{4\tau^2}{3}\right). 
\eqa
     \item under the general alternative $H_a^G$we have
\bqa
\left(\frac{n}{k_n}\right)^{1/2}T_n \dkonv N\left(\text{Var}(\tilde{m}_\ell(X_{\ell t})) + Var(\tilde{m}(\vX)) , \frac{4\tau^2}{3}\right). 
\eqa
\end{enumerate}
\end{theorem}
The proof of Theorem \ref{thm.altern} can be found in the appendix.
\begin{remark} 
If the additional effect $(nk_n)^{-1/4}\tilde{m}_\ell (X_{\ell t})$ approaches zero at a faster rate than $(nk_n)^{-1/4}$, the asymptotic mean of the test statistic under local alternatives converges to zero. Conversely, if it decreases at a slower rate than $(nk_n)^{-1/4}$, the asymptotic mean diverges to infinity. Similarly for the general alternative.
\end{remark}

\section{Test-Based Variable Selection} \label{sec.variable_selection}

Let 
$I_d = \{1,2,\ldots,d\}$ 
denote the index set corresponding to all available exogenous covariates at time $t$ $\vX_t = (\vX_{1t}, \ldots, \vX_{dt})$. Among these, assume that only a subset is truly relevant for explaining the dynamics of the ARCH-$m(X)$ model. We denote the set of indices corresponding to the true covariates by $I_0 \subseteq I_d$, where $|I_0| = d_0 \leq d$.
The objective of variable selection is to recover $I_0$ from the observed data. Specifically, a variable selection procedure produces an estimated index set
$\widehat{I} \subseteq I_d$, which
ideally, coincides with the true set \( I_0 \), either exactly or asymptotically as the sample size increases. This property is defined as the consistency in variable selection, which more formally can be written as
\bqa
\mathbb{P}\!\left(\widehat{I} = I_0 \right) \to 1 \quad \text{as } n \to \infty.
\eqa

Let $\vX^I_t$ denote the subset of the vector of covariates with indices in the set $I$ at time $t$. Suppose that the nonparametric function $m(\cdot)$ in the ARCH-m(X) model in (\ref{eq:ARCHmodel}) is truly a function of the exogenous covariates in $I_0$, i.e., 
\bqa
m(\vX_t)=m\left(\vX^{I_0}_t\right),
\eqa
where, with an abuse of notation, the number of arguments of the function $m$ is determined from the dimension of the vector $I$ it is applied to. 
For each covariate index \( \ell \in I_d = \{1,\ldots,d\} \), define the 
the null hypothesis that it does not contribute to the model as
\begin{equation} \label{H0j}
H_{0}^{\ell}: m\left(\mathbf{X}_t\right) = m_\ell\left(\mathbf{X}_t^{-\ell}\right), \;\;\;\;\; \ell = 1, \ldots, d,
\end{equation}
set $T_n^{j}$ as the test statistic associated with testing $H_0^\ell$ as defined in equation (\ref{eq:Tn}), and let 
\bqa
\pi_\ell = 1 - \Phi\left(\left(\frac{n}{k_n}\right)^{1/2}\left(4\hat{\tau}^2/3\right)^{-1/2}T_n^{\ell}\right), \quad \ell = 1, \ldots, d,
\eqa
be the corresponding $p$-value, where $\Phi(\cdot)$ denotes the standard normal distribution function.

Denote by
\[
\pi_{(1)} \leq \pi_{(2)} \leq \cdots \leq \pi_{(d)}
\]
the ordered \(p\)-values, and let \(H_0^{(\ell)}\) be the null hypothesis corresponding to \(\pi_{(\ell)}\), \(\ell=1,\ldots,d\).  
To control the false discovery rate (FDR), we apply the procedure of Benjamini and Hochberg (1995) and generalized by Benjamini and Yekutieli (2001). Specifically, for a target level \(q \in (0,1)\), define
\begin{equation} \label{k.of.fdr}
k = \max_{\ell\in I_d} \left\{ \ell : \pi_{(\ell)} \leq \frac{\ell}{d} \cdot \frac{q}{\sum_{l=1}^{d} l^{-1}} \right\}.
\end{equation}
If such a maximum exists, all hypotheses \(H_0^{(1)}, \ldots, H_0^{(k)}\) are rejected. If no such \(k\) exists, none of the hypotheses are rejected.  
As a result, the proposed variable selection procedure yields the estimator $\hat{I}$ which is composed by the selected index set
\[
\widehat{I} = \{ \ell \in I_d : H_0^{(\ell)} \text{ is rejected} \},
\]
that is, the indices corresponding to the \(k\) smallest \(p\)-values.

Let $R$ denote the total number of rejected null hypotheses, that is, $R = k$ if such a value exists according to equation (\ref{k.of.fdr}), otherwise, $R = 0$. Let $V$ represent the number of erroneously rejected hypotheses so that $Q = 
V/R$
is the proportion of false rejections (if $R = 0$ then $Q = 0$). The false discovery rate is
\[
\mathrm{FDR} := \mathbb{E}(Q) \leq \alpha_n \frac{d - d_0}{d} \leq \alpha_n.
\]

The consistency result presented here allows the statistical significance of the predictors to diminish as the sample size $n$ increases. Lemma~\ref{lemma:A1} establishes key asymptotic properties of the p-values, which are subsequently utilized in Lemma~\ref{lemma:epsilon} and Theorem~\ref{thm.consistency}. In particular, Lemma~\ref{lemma:epsilon} shows that, with probability approaching one, the smallest $d_0$ p-values correspond precisely to the truly relevant predictors. This fact is central to the proof of Theorem~\ref{thm.consistency}, which establishes the consistency of the proposed method.

\begin{lemma} \label{lemma:A1} Assume that Conditions C1-C4 are satisfied. Then
\ben \item[a)] For $\ell \notin I_0$ and any $\xi > 0$, we have $P(\pi_\ell \leq \xi) = \xi + o(1)$.
\item[b)] For $\ell \in I_0$, let $\xi_n> 0, \xi_n \to 0$, $n\ge 1$.  Denote $M_\ell=\max_t\{\tilde{m}_{\ell}(X_{\ell t})\}$. For general alternatives of the form $m(\vX_t) = m_\ell(\vX_t^{-\ell}) + \tilde{m}_\ell(X_{\ell t})$, if
\bqa
\left(\frac{n}{k_n}\right)^{1/2}M_\ell \to \infty, \ \mbox{ and }\   \xi_n  > O\left(\frac{\exp\left\lbrace-\left(\frac{n}{k_n}M_\ell^2\right)/4\right\rbrace}{\sqrt{\frac{n}{k_n}}M_\ell}\right)
\eqa
for some constant $c$, then $P(\pi_\ell > \xi_n) = o(1)$.
\een
\end{lemma}
\begin{proof}
a) Let $\Phi(x)$ denote the cumulative distribution function of the Standard Normal distribution at $x$ and note that the p-value $\pi_\ell= 1 - \Phi\left(\left(\frac{n}{k_n}\right)^{1/2}\left(4\hat{\tau}^2/3\right)^{-1/2}T_n^{\ell}\right)$. Under $H_0^\ell$, $\ell \notin I_0$ and
\bqa
P(\pi_\ell \leq \xi) &=& P\left(1 - \Phi\left(\left(\frac{n}{k_n}\right)^{1/2}\left(4\hat{\tau}^2/3\right)^{-1/2}T_n^{\ell}\right) \leq \xi\right)\\
&=& P\left(\Phi\left(\left(\frac{n}{k_n}\right)^{1/2}\left(4\hat{\tau}^2/3\right)^{-1/2}T_n^{\ell}\right) \geq 1 - \xi\right)\\
&=&P\left(\left(\frac{n}{k_n}\right)^{1/2}\left(4\hat{\tau}^2/3\right)^{-1/2}T_n^{\ell} \geq \Phi^{-1}\left(1 - \xi\right)\right) \leq \xi + o(1),
\eqa
where the inequality follows from Theorem \ref{thm.null} under $H_0^\ell$ and the fact that $\hat{\tau}$ is a consistent estimator of $\tau$.

b) For $\ell \in I_0$ note that, the standardized test statistic under local additive alternatives where $m(\vX_t) = m_\ell(\vX_t^{-\ell}) + \tilde{m}_\ell(X_{\ell t})$  (similarly for general alternatives) can be written as
\bqa
\left(\frac{n}{k_n}\right)^{1/2}\left(4\hat{\tau}^2/3\right)^{-1/2}T_n^{\ell} &=& \left(\frac{n}{k_n}\right)^{1/2}\left(4\hat{\tau}^2/3\right)^{-1/2}\vV'\vA\vV \\
&=& \left(\frac{n}{k_n}\right)^{1/2}\left(4\hat{\tau}^2/3\right)^{-1/2}(\vv_V + \valpha_V + \vm_V)'\vA(\vv_V + \valpha_V + \vm_V) \\
&&+  2\left(\frac{n}{k_n}\right)^{1/2}\left(4\hat{\tau}^2/3\right)^{-1/2}(\vv_V + \valpha_V + \vm_V)'\vA\tilde{\vm}_{V\ell}\nonumber\\
&&+ \left(\frac{n}{k_n}\right)^{1/2}\left(4\hat{\tau}^2/3\right)^{-1/2}\tilde{\vm}_{V\ell}'\vA\tilde{\vm}_{V\ell}\nonumber\\
&=& Z_\ell^{H_0^\ell} + O_p\left(\left(\frac{n}{k_n}\right)^{1/2}M_\ell\right)
\eqa
where $\tilde{\vm}_{V\ell}=(\tilde{m}_\ell(X_{\ell t}), t\in W_1, \cdots, \tilde{m}_\ell(X_{\ell t}), t\in W_n)$, $\vV = (\hat{v}_t, t \in W_1, \ldots, \hat{v}_t, t \in W_n)$, $\valpha_V = \Big((\hat{\alpha}_1 - \alpha_1)\epsilon_{t-1}^2, t \in W_1, \ldots,$ $(\hat{\alpha}_1 - \alpha_1)\epsilon_{t-1}^2, t\in W_n\Big)$, $\vv_V = (v_t,t \in W_1, \ldots, v_t, t \in W_n)$, and  $\vm_V = (({\hat{\vtheta}}^{-\ell}-\vtheta^{-\ell})^T \mathcal{B}(\vX_{t-1}^{-\ell}), t\in W_1, \ldots, ({\hat{\vtheta}}^{-\ell}-\vtheta^{-\ell})^T \mathcal{B}(\vX_{t-1}^{-\ell}), t\in W_n)$.  Note that by Theorem \ref{thm.null}, $Z_\ell^{H_0^\ell}$ converges in distribution to a $N(0,1)$. Let $a_n \to \infty$ such that $a_n = o\left(\left(\frac{n}{k_n}\right)^{1/2}M_\ell\right)$. Then,
\bqa
P(\pi_k > \xi_n) &=& P\left(1 - \Phi\left(\left(\frac{n}{k_n}\right)^{1/2}\left(4\hat{\tau}^2/3\right)^{-1/2}T_n^{\ell}\right) > \xi_n\right) \\
 &=& P\!\left(1 - \Phi\!\left(Z^{H^\ell_0}_\ell + O_p\left(\left(\frac{n}{k_n}\right)^{1/2}M_\ell\right)\right) > \xi_n\right) \\
&\leq& P\!\left(1 - \Phi\!\left(Z^{H^\ell_0}_\ell + O_p\left(\left(\frac{n}{k_n}\right)^{1/2}M_\ell\right)\right) > \xi_n,\; Z^{H^\ell_0}_\ell\geq -a_n\right)  + P\!\left(Z^{H^\ell_0}_\ell < -a_n\right) \\
&\leq& P\!\left(1 - \Phi\!\left(O_p\left(\left(\frac{n}{k_n}\right)^{1/2}M_\ell\right) - a_n\right) > \xi_n\right) + o(1) \;=\; o(1),
\eqa
For the last line, we use the fact that $P\!\left(1 - \Phi\!\left(O_p\left(\left(\frac{n}{k_n}\right)^{1/2}M_\ell\right) - a_n\right) > \xi_n\right)\to 0$ as $n\to \infty$. The explanation is as follows: because $a_n = o\left(\left(\frac{n}{k_n}\right)^{1/2}M_\ell\right)$, it can be shown that 
\[
1-\Phi\left(O_p\left((n/k_n)^{1/2}M_\ell\right) -a_n\right)\le O\left(\exp\left\lbrace-\left((n/k_n)M_\ell^2\right)/4\right\rbrace/\sqrt{n/k_n}M_\ell\right).
\] 
Thus, if $\xi_n$ is set to be $\xi_n > O\left(\exp\left\lbrace-\left((n/k_n)M_\ell^2\right)/4\right\rbrace/\sqrt{n/k_n}M_\ell\right)$, the desired conclusion follows immediately.



\end{proof}

\begin{lemma} \label{lemma:epsilon} Let ${\cal E}_n$ denote the event where the smallest $d_0$ p-values are the p-values corresponding to the $d_0$ significant covariates, i.e.,  with $I_0 = \{j_1,\ldots,j_{d_0}\}$,
\bqa
{\cal E}_n = \big[\{\pi_{(1)},\ldots,\pi_{(d_0)}\} = \{\pi_{j_1},\ldots,\pi_{j_{d_0}}\}\big].
\eqa
Then, under assumptions A1-A11 in \cite{FrancqThieu2019}, we have that 
\bqa
\lim_{n\rightarrow \infty}P({\cal E}_n) = 1.
\eqa
\end{lemma}
\begin{proof}
Let $\xi$ be any number between 0 and 1, and write
\bqa
P({\cal E}_n^c) & \leq & \sum_{\ell_2 \in I_0}\sum_{\ell_1 \notin I_0}P(\pi_{\ell_1} < \pi_{\ell_2}) \nonumber \\
&=& \sum_{\ell_2 \in I_0}\sum_{\ell_1 \notin I_0}\big[P([\pi_{\ell_1} < \pi_{\ell_2}] \cap [\pi_{\ell_1} \leq \xi])  + P([\pi_{\ell_1} < \pi_{\ell_2}] \cap [\pi_{\ell_1} > \xi])\big]\nonumber \\
&\leq& \sum_{\ell_2 \in I_0}\sum_{\ell_1 \notin I_0}\left[P(\pi_{\ell_1}\leq \xi)  + P(\pi_{\ell_2} > \xi)\right]
 \leq  \sum_{\ell_2 \in I_0}\sum_{\ell_1 \notin I_0}\left[\xi + o(1)\right] \ \ \mbox{(by Lemma \ref{lemma:A1})} \nonumber \\
&=& d_0(d - d_0)\xi + o(1).  \nonumber
\eqa
Since $\xi$ is arbitrary, this shows that $\lim_{n \rightarrow \infty}P({\cal E}_n^c) = 0$, completing the proof.
\end{proof}

\begin{thm} \label{thm.consistency}
With $\alpha_n$ the chosen bound of FDR (or Bonferroni), assume that $\alpha_n \to 0$ as $n \to \infty$ in such a way that
\bqa
\alpha_n > O\left(\frac{d}{d_0}\left(\sum_{\ell=1}^d\ell^{-1}\right)\frac{\exp\left\lbrace-\left(\frac{n}{k_n}M_\ell^2\right)/4\right\rbrace}{\sqrt{\frac{n}{k_n}}M_\ell}\right).
\eqa
Then, under assumptions A1-A11 in \cite{FrancqThieu2019}, $$\lim_{n\to \infty}P(\hat{I} = I_0) = 1.$$
\end{thm}

\begin{proof} 
Note that if the estimator $\hat{I}$ is equal to the set $I_0$, we have exactly $d_0$ rejections ($R = d_0$) with none of them being erroneous ($V = 0$). Therefore, consistency of $\hat{I}$ is verified by proving
\begin{equation}
P(\hat{I} = I_0) = P(R = d_0, V = 0) \rightarrow 1,\  \mbox{as}\  n \rightarrow \infty.
\end{equation}
This follows by showing that both $P(R \neq d_0)$ and $P(V \geq 1)$ are asymptotic negligible. By Lemma 2.1 in \cite{BuneaEtAl2006}, we have that 
\bqa
P(V \geq 1) \leq  P(R \neq d_0) + \frac{d_0(d-d_0)}{d}\alpha_n = o(1)
\eqa
by the choice of $\alpha_n \to 0$. as long as $P(R \neq d_0) \rightarrow 0$. Following \cite{BuneaEtAl2006}, consider the two cases: $\{R > d_0\}$ and $\{R < d_0\}$.

 For $\{R > d_0\}$, at least one false positive is made. Then
    \begin{align}
        \exists \ell_0 \not \in I_0 \text{ such that } \pi_{(\ell_0)} \leq \frac{\ell_0}{d}\frac{\alpha_n}{\sum_{l=1}^{d}l^{-1}}.
        \notag
    \end{align}
    This event can be written as
    \begin{align}
        \bigg\{\bigcup_{\ell=d_0+1}^{d}\bigg\{\pi_{(\ell)} \leq \frac{\ell}{d}\frac{\alpha_n}{\sum_{l=1}^{d}l^{-1}}\bigg\}\bigg\}.
        \notag
    \end{align}
    For $\{R < d_0\}$, there is at least one false negative.  Under FDR, this can be expressed  as
    \begin{align}
         \bigg\{\pi_{(d_0)} > \frac{d_0}{d}\frac{\alpha_n}{\sum_{l=1}^{d}l^{-1}}\bigg\}
        \notag
    \end{align}
Consequently, the probability of the event $\{R \neq d_0\}$ can be bounded from above as follows:
    \begin{align} \label{eq.R}
        P(\{R \neq d_0\}) &\leq P\bigg({\cal E}_n\bigcap\bigg\{\pi_{(d_0)} > \frac{d_0}{d}\frac{\alpha_n}{\sum_{l=1}^{d}l^{-1}}\bigg\}\bigg)  \nonumber\\ 
        &\hspace{.5cm} + \sum_{k=d_0+1}^{d}P\bigg({\cal E}_n\bigcap\bigg\{\pi_{(k)} \leq \frac{k}{d}\frac{\alpha_n}{\sum_{l=1}^{d}l^{-1}}\bigg\}\bigg) + P({\cal E}_n^C)
    \end{align}
    Appling Lemma \ref{lemma:epsilon}, it remains to show that the first two terms on the right hand side of equation \eqref{eq.R} go to 0 as the sample size increases.  
    
    By Lemma \ref{lemma:A1} and Lemma \ref{lemma:epsilon}, the first term on the right hand side of equation \eqref{eq.R} can be bounded from above by 
    \begin{align}
        P\left({\cal E}_n\bigcap\left\{\pi_{(d_0)} > \frac{d_0}{d}\frac{\alpha_n}{\sum_{l=1}^{d}l^{-1}}\right\}\right) &\leq d_0\cdot \text{max}_{k\in I_0}P\left(\pi_k > \frac{d_0}{d}\frac{\alpha_n}{\sum_{l=1}^{d}l^{-1}}\right)
        \notag \\
        &= d_0\cdot\text{max}_{k\in I_0}P\left(1 - \Phi\left(\left(\frac{n}{k_n}\right)^{1/2}\left(4\hat{\tau}^2/3\right)^{-1/2}T_n^{\ell}\right) > \frac{d_0}{d}\frac{\alpha_n}{\sum_{l=1}^{d}l^{-1}}\right)
        \notag \\
        &= d_0\cdot\text{max}_{k\in I_0}P\left(\left(\frac{n}{k_n}\right)^{1/2}\left(4\hat{\tau}^2/3\right)^{-1/2}T_n^{\ell} < \Phi^{-1}\left(1 - \frac{d_0}{d}\frac{\alpha_n}{\sum_{l=1}^{d}l^{-1}}\right)\right)
        \notag \\
        &= o(1) \text{ as } n \to \infty
        \notag
    \end{align}
    
    For the second term on the right hand side of (\ref{eq.R}) we have
    \begin{align}
    \sum_{k=d_0+1}^{d}P\bigg({\cal E}_n\bigcap\bigg\{\pi_{(k)} \leq \frac{k}{d}\frac{\alpha_n}{\sum_{l=1}^{d}l^{-1}}\bigg\}\bigg) &\leq \sum_{k=d_0+1}^{d}P\bigg({\cal E}_n\bigcap\bigg\{\pi_{(k)} \leq \frac{\alpha_n}{\sum_{l=1}^{d}l^{-1}}\bigg\}\bigg)
        \notag \\
        &\leq \sum_{k\not \in I_0}P\bigg(\pi_k \leq \frac{\alpha_n}{\sum_{l=1}^{d}l^{-1}}\bigg)
        \notag \\
        &= o(1) \text{ as } n \to \infty,
        \notag
    \end{align}
    since $\pi_k$ follows a Uniform distribution under the null hypothesis $H_0^k$ ($k\notin I_0$).
    
    The proof under the Bonferroni correction follows similarly by noting that the cut-off is $\alpha_n/d$ instead of $\alpha_n/\sum_{l=1}^{d}l^{-1}$.

\end{proof}

\section{Simulations} \label{sec.simulations}

\subsection{Nonparametric Testing}

We first consider a simulation study to identify significant exogenous covariates in an ARCH-$m(\mathbf{X})$ model using the proposed nonparametric hypothesis testing procedure. Implementing the procedure discussed in Section \ref{sec:method} becomes computationally demanding as the number of covariates increases. Specifically, in the B-spline method, each dimension of the exogenous covariate $X_\ell\ (1\leq \ell\leq d)$ is associated with $K_\ell$ parameters $\boldsymbol{\theta}$ in equation (\ref{m_param}), corresponding to the $K_\ell$ basis functions. This estimation process is challenging because optimization in a high-dimensional space is inherently difficult; it increases the risk of convergence to local optimal, potentially failing to reach the true global solution. Consequently, empirical parameter estimation in multidimensional settings requires an extensive number of observations. To overcome this hurdle, we introduce a modified kernel regression approach, adapted from the partially linear model framework to the autoregressive form in equation (\ref{eq:model-new}). We discuss its detailed formulation below.

Given the true parameters $\boldsymbol{\alpha}=(\alpha_1, \alpha_2, \ldots, \alpha_p)$, we have 
\bqa
E(\epsilon_t^2 - \sum_{j=1}^p\alpha_j\epsilon_{t-j}^2) = m(\vX_{t-1})
\eqa
so that a natural estimator of $m(\cdot)$ is the kernel smoothing estimator 
\bqan \label{eq.m_hat}
\hat{m}(\mathbf{x}_{t-1}) &=& \hat{m}(\mathbf{x}_{t-1},\valpha)=\sum_{\ell=p+1}^{n}W_{\ell}(\mathbf{x}_{t-1})\left(\epsilon_\ell^2-\sum_{j=1}^p\alpha_j\epsilon_{\ell-j}^2\right)\\
&=& \begin{bmatrix}
W_{p+1}(\vx_{t-1}) & W_{p+2}(\vx_{t-1}) & \ldots & W_{n}(\vx_{t-1})
\end{bmatrix}\left(\begin{bmatrix}
\epsilon_{p+1}^2\\
\vdots\\
\epsilon_{n}^2
\end{bmatrix} 
-
\begin{bmatrix}
\epsilon_{p+1-1}^2 & \epsilon_{p+1-2}^2 & ... & \epsilon_{p+1-p}^2\\
\epsilon_{p+2-1}^2 & \epsilon_{p+2-2}^2 & ... & \epsilon_{p+2-p}^2\\
\vdots\\
\epsilon_{n-1}^2 & \epsilon_{n-2}^2 & ... & \epsilon_{n-p}^2\\
\end{bmatrix}\begin{bmatrix}
\alpha_{1}\\ \vdots \\ \alpha_p
\end{bmatrix}
\right),\nonumber
\eqan
where
$$W_\ell(\mathbf{x}_{t-1})=\frac{K_h(\mathbf{x}_{t-1}-\mathbf{x}_{\ell-1})}{\sum_{m=p+1}^{n}K_h(\mathbf{x}_{t-1}-\mathbf{x}_{m-1})}$$
and $K_h(x) := K(x/h)$ is a symmetric kernel density.

From (\ref{eq:ARCHmodel}) and (\ref{eq.m_hat}) we can define the estimator of $\valpha$ as a solution to the Least Squares Criterion, where the objective function RSS for minimization can be written in the following form.
\bqa
\text{RSS}&=&\sum_{t=p+1}^n\hat{v}_t^2
=\sum_{t=p+1}^n\left[\epsilon_t^2-\sum_{j=1}^p\alpha_j\epsilon_{t-j}^2-\hat{m}(\mathbf{x}_{t-1})\right]^2\\
&=&\sum_{t = 1}^n\left(\epsilon_t^2 - \sum_{j=1}^p\alpha_j\epsilon_{t-j}^2 -  \sum_{s=1}^n W_{s}(\vX_{t-1})\left(\epsilon_s^2 -  \sum_{j = 1}^p\alpha_j\epsilon_{s-j}^2\right)\right)^2\\
&=&\sum_{t = 1}^n\left(\epsilon_t^2 - \sum_{s=1}^n W_{s}(\vX_{t-1})\epsilon_s^2    - \sum_{j=1}^p\alpha_j\epsilon_{t-j}^2 + \sum_{s=1}^n W_{s}(\vX_{t-1})\left(\sum_{j = 1}^p\alpha_j\epsilon_{s-j}^2\right) \right)^2\nonumber\\
&=&  \sum_{t = 1}^n\left(\tilde{\epsilon}^2_t  - \sum_{j=1}^p\alpha_j\left(\epsilon_{t-j}^2 - \sum_{s=1}^n W_{s}(\vX_{t-1})\epsilon^2_{s-j}\right) \right)^2 \nonumber\\
&=&\sum_{t = 1}^n\left(\tilde{\epsilon}^2_{t,0} - 
 \sum_{j=1}^p\alpha_j\tilde{\epsilon}_{t,j}^2 \right)^2
\eqa
where $\tilde{\epsilon}_{t,k}^2=\epsilon_{t-k}^2-\sum_{\ell=p+1}^TW_\ell(\mathbf{x}_{t-1})\epsilon_{\ell-k}^2$. 

Define
$
\tilde{\vY} = \begin{bmatrix}
\tilde{\epsilon}_{p+1,0}^2 & 
\tilde{\epsilon}_{p+2,0}^2 & 
\ldots & 
\tilde{\epsilon}_{n,0}^2
\end{bmatrix}^T = \vepsilon_{p+1} - \vW\vepsilon_{p+1}
$
and 
\bqa
\tilde{\vX} &=& \begin{bmatrix}
\tilde{\epsilon}_{p+1,1}^2 & \tilde{\epsilon}_{p+1,2}^2 & \cdots & \tilde{\epsilon}_{p+1,p}^2\\
\tilde{\epsilon}_{p+2,1}^2 & \tilde{\epsilon}_{p+2,2}^2 & \cdots & \tilde{\epsilon}_{p+2,p}^2\\
\vdots & \vdots & \ddots & \vdots\\
\tilde{\epsilon}_{n,1}^2 & \tilde{\epsilon}_{n,2}^2 & \cdots & \tilde{\epsilon}_{n,p}^2
\end{bmatrix} = \begin{bmatrix}
\vepsilon_p & \vepsilon_{p-1} & \ldots & \vepsilon_{1}
\end{bmatrix}
-
\vW\begin{bmatrix}
\vepsilon_p & \vepsilon_{p-1} & \ldots & \vepsilon_{1}
\end{bmatrix},
\eqa
where
\bqa
\vepsilon_\ell =  \begin{bmatrix}
\epsilon_\ell\\
\epsilon_{\ell+1}\\
\vdots\\
\epsilon_{n-(p+\ell-1)}
\end{bmatrix}
\mbox{ and } \vW = \begin{bmatrix}
W_{p+1}(\vx_{p+1-1}) & W_{p+2}(\vx_{p+1-1}) & ... & W_{n}(\vx_{p+1-1})\\
W_{p+1}(\vx_{p+2-1}) & W_{p+2}(\vx_{p+2-1}) & ... & W_{n}(\vx_{p+2-1})\\
...\\
W_{p+1}(\vx_{n-1}) & W_{p+2}(\vx_{n-1}) & ... & W_{n}(\vx_{n-1})\\
\end{bmatrix}.
\eqa

In addition, define the matrix of kernels as
$$
\vK_h = \begin{bmatrix}
K_h(\vx_{p+1} - \vx_{p+1}) & K_h(\vx_{p+2} - \vx_{p+1}) & \ldots & K_h(\vx_{n} - \vx_{p+1})\\
K_h(\vx_{p+1} - \vx_{p+2}) & K_h(\vx_{p+2} - \vx_{p+2}) & \ldots & K_h(\vx_{n} - \vx_{p+2})\\
\vdots\\
K_h(\vx_{p+1} - \vx_{n}) & K_h(\vx_{p+2} - \vx_{n}) & \ldots & K_h(\vx_{n} - \vx_{n})\\
\end{bmatrix}
$$
and note that $
\vW = \text{diag}(\vK\vone)^{-1}\vK
$
where $\vone = (1,1,1,\ldots, 1)^T$.
Then, RSS can be written as
\bqa
\text{RSS} = \sum_{t=p+1}^n\left(\tilde{\epsilon}_{t,0}^2-\sum_{k=1}^p\alpha_k\tilde{\epsilon}_{t,k}^2\right)^2 
= \|\tilde{\mathbf{Y}}-\tilde{\mathbf{X}}\alpha\|^2
\eqa
so that the Least Squares solution for $\valpha$ is
\bqa
\hat{\valpha} = \arg\min_{\valpha} || \tilde{\vY} - \tilde{\vX}\valpha||^2 = (\tilde{\vX}'\tilde{\vX})^{-1}\tilde{\vX}\tilde{\vY}.
\eqa
After the parameter estimates $\hat{\valpha} := (\hat{\alpha}_1, \ldots, \hat{\alpha}_p)$ are obtained, they can be used to estimate $m(\cdot)$ as in equation (\ref{eq.m_hat}).

In what follows, we present comprehensive simulation studies that compare the performance of the proposed nonparametric testing procedure with the benchmark method in \cite{FrancqThieu2019}, where the latter assumes a linear functional form for $m(\mathbf{x})$ in the model defined in Equation \eqref{eq:ARCHmodel}. Specifically, we consider two scenarios: one where $m(\mathbf{x})$ is a statistical functional of 2 covariates and another where there are 5 covariates available. In addition, we assess the performance of the tests using eight different models in each setting, where the model can be a linear, quadratic, or a non-linear function of the covariates in either an additive or non-additive form. In particular, Model 7 in each scenario assumes a linear form for $m(\mathbf{x})$, which agrees with the underlying assumption for the benchmark method by \cite{FrancqThieu2019}. Therefore, one would expect the benchmark to outperform the proposed method in that case. 

All sixteen models under consideration are summarized in Table \ref{tab:models}. In the table, $\epsilon_t=\sigma_t z_t$, $\sigma_t^2 = m_t(\mathbf{x})$, and the parameter $c$ takes values from $\{0,0.1,0.2,\ldots,1\}$, which reflects the strength of the signal for the active covariates. Given a 2-covariate model, the null hypothesis of interest is $H_0: m(\mathbf{x})=m_2(\mathbf{x}^{-2}) = m_2(x_1)$. In contrast, the null hypothesis for each 5-covariate model is $H_0: m(\mathbf{x})=m_3(\mathbf{x}^{-3})$. In all cases the covariates are generated with a Multivariate Normal distribution with exponential decay covariance, that is, $\vX_t \sim N(2.5,\Sigma)$, where $\Sigma_{ij} = \rho^{-|i-j|}$, with $\rho = 0$ or $\rho = 0.5$.

We generated 500 independent samples of size $n=500\text{ or }1000$ in each scenario and computed the percent of times the null hypothesis is rejected at a significance level 0.05. We considered three different distributions for the shock $z$, namely, the Standard Normal distribution, the Laplace distribution with a scale parameter of 0.5, and a $t$ distribution with 7 degrees of freedom. In the 2-covariate setting, the proposed method was evaluated using both the B-spline and the kernel smoothing estimation, while for the 5-covariate case, due to the challenge in high-dimensional parameter estimation, only kernel smoothing was used. Note that when $c=0$, the calculated relative frequency is the empirical level of the test, while it represents the statistical power of the test for $c>0$.

Figures \ref{fig2}-\ref{fig14} display the comparison of the rejection rates between the proposal and the benchmark for sample size $n=1000$ and $\rho=0$ with the Standard Normal shock distribution. 
Figures 1-24 in the Supplementary Material display the comparison for all scenarios considered. 

\begin{table}[http]
\begin{center}
\caption{Models under Consideration in Nonparametric Testing}\label{tab:models}
{\small
\begin{tabular}{c|l|l}
\hline
Model & Parameter & \multicolumn{1}{c}{2-Covariate Scenario}\\
\hline
1& $\valpha = 0.3$ & $m(\mathbf{x})=0.2 + 0.3\epsilon_{t-1}^\delta + (1/2)(x_{1, t-1}- 2.5)^2 +2c(x_{2,t-1} - 2.5)^2$\\    
2 & $\valpha = 0.3$ &  $m(\mathbf{x})=0.2 +  0.3\epsilon_{t-1}^\delta + (1/2)(x_{1,t-1}+ 2cx_{2,t-1}- 10)^2$\\
3 &  $\valpha = (0.3, 0.2)$ &  $m(\mathbf{x})=0.2 + 0.3\epsilon_{t-1}^\delta + 0.2\epsilon_{t-2}^\delta + (1/2)(x_{1,t-1} - 2.5)^2 + 2c(x_{2,t-1} - 2.5)^2$ \\  
4 & $\valpha = (0.3, 0.2)$ &  $m(\mathbf{x})=0.2 +  0.3\epsilon_{t-1}^\delta + 0.2\epsilon_{t-2}^\delta + (1/2)(x_{1,t-1} + 2cx_{2,t-1} - 10)^2$ \\
5 &  $\valpha = 0.4$ & $m(\mathbf{x})=0.2 +  0.4\epsilon_{t-1}^\delta  + \{\sin(\pi x_{1,t-1}/2) \}^2+ 10c\{\sin(\pi x_{2,t-1}/2)\}^2$\\   
6 &   $\valpha = 0.4$ & $m(\mathbf{x})=0.2 +  0.4\epsilon_{t-1}^\delta  + 10c\{\sin(x_{1,t-1}/2+\pi c x_{2,t-1}/2)\}^2$ \\   
7 &   $\valpha = 0.4$ & $m(\mathbf{x})=0.2+0.4\epsilon_{t-1}^\delta + (1/4)x_{1,t-1} + cx_{2,t-1}$ \\
8 &   $\valpha = 0.4$ & $m(\mathbf{x})=0.2 +  0.4\epsilon_{t-1}^\delta + cx_{1,t-1}x_{2,t-1}$ \\
 \hline
Model & Parameter & \multicolumn{1}{c}{5-Covariate Scenario}\\
\hline
1&  $\valpha = 0.3$ &  $m(\mathbf{x})=0.2 +  0.3\epsilon_{t-1}^\delta + (1/2)\sum_{i\ne 3}(x_{i,t-1} - 2.5)^2 + 5c(x_{3,t-1}-2.5)^2$ \\
2 &  $\valpha = 0.3$ &  $m(\mathbf{x})=0.2+0.3\epsilon_{t-1}^\delta + (1/2)(x_{1,t-1} + x_{2,t-1} + 5c x_{3,t-1} + x_{4,t-1} + x_{5,t-1} - 25)^2$  \\
3&   $\valpha = (0.3, 0.2)$ & $m(\mathbf{x})=0.2+  0.3\epsilon_{t-1}^\delta + 0.2\epsilon_{t-2}^\delta + (1/2)\sum_{i\ne 3}(x_{i,t-1} - 2.5)^2 + 5c (x_{3,t-1}-2.5)^2 $\\
4 &   $\valpha = (0.3, 0.2)$ & $m(\mathbf{x})=0.2+0.3\epsilon_{t-1}^\delta + 0.2\epsilon_{t-2}^\delta + (1/2)(x_{1,t-1}+ x_{2,t-1} + 5c x_{3,t-1} + x_{4,t-1} + x_{5,t-1}  - 25)^2$  \\
5 &   $\valpha = 0.4$ & $m(\mathbf{x})=0.2+  0.4\epsilon_{t-1}^\delta  + \sum_{i\ne 3}\sin(\pi x_{i,t-1}/2) +10c\sin(\pi x_{3,t-1}/2)$ \\
6 &  $\valpha = 0.4$ &  $m(\mathbf{x})=0.2 + 0.4\epsilon_{t-1}^\delta +\{\sin(\pi x_{t-1,1}/2)\}^2\{\sin(\pi x_{t-1,2}/2)\}^2$\\
&& $+  10c\{\sin(\pi x_{t-1,3}/2)\}^2\{\sin(\pi x_{t-1,4}/2)\}^2 + \{\sin(\pi x_{t-1,5}/2)\}^2$  \\
7 &  $\valpha = 0.4$ &  $m(\mathbf{x})=0.2+  0.4\epsilon_{t-1}^\delta + (1/4)\sum_{i\ne 3}x_{i,t-1}+ 8cx_{3,t-1}$ \\
8 &   $\valpha = 0.4$ & $m(\mathbf{x})=0.2+  0.4\epsilon_{t-1}^\delta + cx_{t-1,3}\sum_{i\ne3}x_{t-1,i}$ \\
\hline
\end{tabular}
}
\end{center}
\end{table}

The simulation results suggest that in the 2-covariate scenarios, the proposed method (implemented via kernel smoothing or B-splines) consistently outperforms the benchmark, achieving significantly higher rejection rates for $c>0$ across Models 1–6. This trend persists across various error distributions and holds for both independent and correlated covariates. While \cite{FrancqThieu2019}'s approach maintains an expected advantage in Models 7 and 8, the proposed approach shows dramatic improvements elsewhere, where the benchmark fails to detect nonlinear functions. Note that Model 6 appears particularly challenging to fit and perform inference for, especially under Laplace errors. Nevertheless, the proposed methodology shows clear advantage over the benchmark in this difficult case. Both kernel and B-splines versions of the proposal yield comparable results. Generally, rejection rates improve with larger sample sizes, and performance is superior under Normal or t error distributions compared to Laplace errors. Furthermore, all methods maintain a nominal size (at $c=0$) relatively close to $0.05$.

In the 5-covariate scenario, the proposed method with kernel estimation demonstrates a significantly higher rejection rate for $c>0$ in Models 1–6. Even in Models 7 and 8, its performance remains competitive with the benchmark under Normal error distributions, regardless of covariate correlation or sample size.

\begin{figure}[http]
\begin{center}
\includegraphics[scale=0.7]{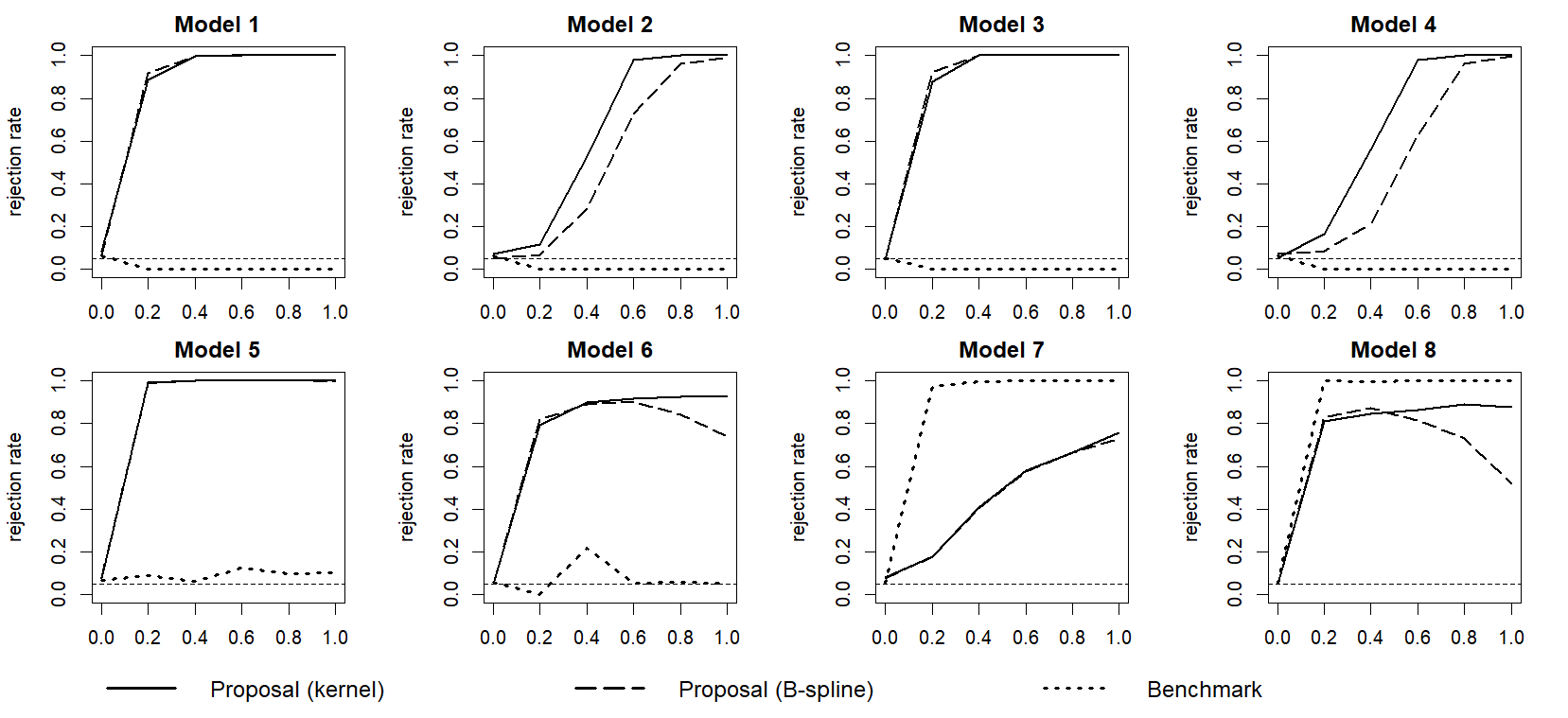}
\caption{Rejection rates against parameter $c$ given a 2-covariate model and sample size $n=1000$ under a normal error distribution with independent covariates ($\rho=0$). The solid curve represents results of the proposal realized by the kernel smoothing method, and the dashed curve corresponds to the benchmark.}\label{fig2}
\end{center}
\end{figure}

\begin{figure}[http]
\begin{center}
\includegraphics[scale=0.7]{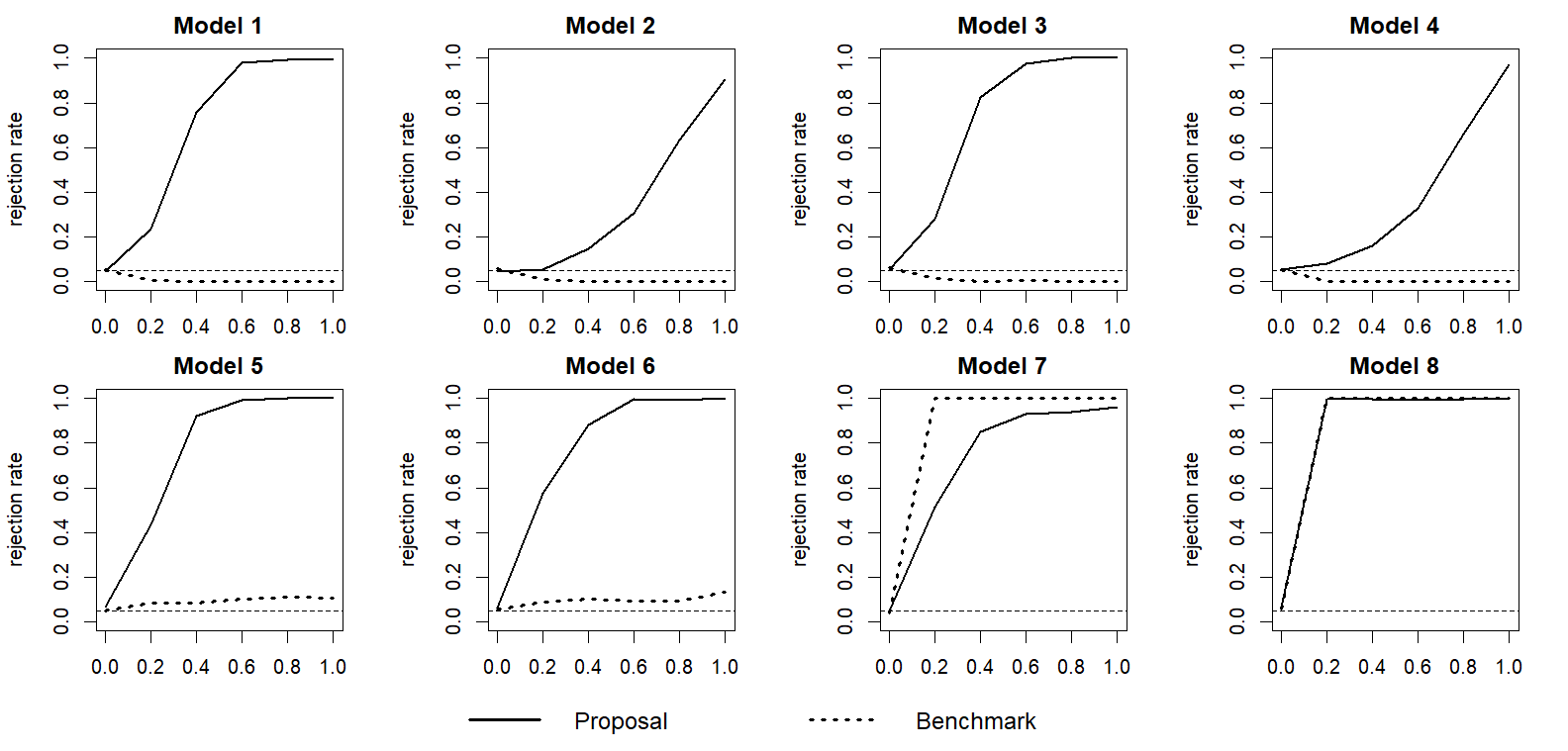}
\caption{Rejection rates against parameter $c$ given a 5-covariate model and sample size $n=1000$ under a normal error distribution with independent covariates ($\rho=0$). The solid curve represents results of the proposal realized by the kernel smoothing method, and the dotted curve corresponds to the benchmark. The horizontal dashed line marks the nominal level of 0.05.}\label{fig14}
\end{center}
\end{figure}

\subsection{Variable Selection}

In this section, we examine the finite-sample performance of the proposed 
variable selection procedure within the ARCH framework under nonparametric 
covariate effects (denoted as ARM), as described in Section~\ref{sec.variable_selection}. 
In parallel with the hypothesis testing experiments, we compare the proposed 
approach with competing methods, namely, the conventional linear method assuming $m(X)$ is of a linear form \citep{FrancqThieu2019} and the Lasso method where $m(X)$ is of a linear form with a Lasso penalty, under three innovation distributions—Standard Normal, 
Student-$t$, and Laplace—considering both $\rho = 0$ and $\rho = 0.5$. The simulation design further varies the number of exogenous covariates  ($d = 5$ and $10$) and the sample size ($n = 500, 1000, 5000,$ and $10000$). Table~\ref{tab:models_d_10} summarizes the sixteen data-generating models examined for both $d = 5$ and $d = 10$.

\begin{table}[http]
\begin{center}
\caption{Models under consideration}\label{tab:models_d_10}
{\small
\begin{tabular}{c|l|l}
\hline
Model &  Parameter & \multicolumn{1}{c}{5-Covariate Scenario}\\
\hline
1 & $\valpha = 0.3$ &
$m(\mathbf{x})=
2 (x_{1,t-1}-2.5)^2
+2 (x_{3,t-1}-2.5)^2
+2 (x_{4,t-1}-2.5)^2$\\

2 & $\valpha = 0.3$ &
$m(\mathbf{x})=
\frac{1}{2}\Big(2 x_{1,t-1}+2 x_{3,t-1}+2 x_{4,t-1}-15\Big)^2$\\

3 & $\valpha = (0.3,\;0.2)$ &
$m(\mathbf{x})=
2 (x_{1,t-1}-2.5)^2
+2 (x_{3,t-1}-2.5)^2
+2 (x_{4,t-1}-2.5)^2$\\

4 & $\valpha = (0.3,\;0.2)$ &
$m(\mathbf{x})=
\frac{1}{2}\Big(2 x_{1,t-1}+2 x_{3,t-1}+2 x_{4,t-1}-15\Big)^2$\\

5 & $\valpha = 0.4$ &
$m(\mathbf{x})=
8 \sin^2\!\Big(\frac{\pi x_{1,t-1}}{2}\Big)
+8 \sin^2\!\Big(\frac{\pi x_{3,t-1}}{2}\Big)
+8 \sin^2\!\Big(\frac{\pi x_{4,t-1}}{2}\Big)$\\

6 & $\valpha = 0.4$ &
$m(\mathbf{x})=
10 \sin^2\!\Big(\frac{\pi x_{1,t-1}}{2}\Big)\sin^2\!\Big(\frac{\pi x_{3,t-1}}{2}\Big)
+10 \sin^2\!\Big(\frac{\pi x_{4,t-1}}{2}\Big)$\\

7 & $\valpha = 0.4$ &
$m(\mathbf{x})=
8 \,x_{1,t-1}+8 \,x_{3,t-1}+8 \,x_{4,t-1}$\\

8 & $\valpha = 0.4$ &
$m(\mathbf{x})=
 \,x_{3,t-1}\big(x_{1,t-1}+x_{4,t-1}\big)$\\
 \hline
 \hline
Model & Parameter & \multicolumn{1}{c}{10-Covariate Scenario}\\
\hline
1 & $\valpha = 0.3$ &
$m(\mathbf{x})=
2(x_{1,t-1}-2.5)^2
+2(x_{3,t-1}-2.5)^2
+2(x_{4,t-1}-2.5)^2
+2(x_{5,t-1}-2.5)^2
+2(x_{9,t-1}-2.5)^2$\\

2 & $\valpha = 0.3$ &
$m(\mathbf{x})=
\frac{1}{2}\Big(
2x_{1,t-1}
+2x_{3,t-1}
+2x_{4,t-1}
+2x_{5,t-1}
+2x_{9,t-1}
-25
\Big)^2$\\

3 & $\valpha = (0.3,\;0.2)$ &
$m(\mathbf{x})=
2(x_{1,t-1}-2.5)^2
+2(x_{3,t-1}-2.5)^2
+2(x_{4,t-1}-2.5)^2
+2(x_{5,t-1}-2.5)^2
+2(x_{9,t-1}-2.5)^2$\\

4 & $\valpha = (0.3,\;0.2)$ &
$m(\mathbf{x})=
\frac{1}{2}\Big(
2x_{1,t-1}
+2x_{3,t-1}
+2x_{4,t-1}
+2x_{5,t-1}
+2x_{9,t-1}
-15
\Big)^2$\\

5 & $\valpha = 0.4$ &
$m(\mathbf{x})=
8\sin^2\!\Big(\frac{\pi x_{1,t-1}}{2}\Big)
+8\sin^2\!\Big(\frac{\pi x_{3,t-1}}{2}\Big)
+8\sin^2\!\Big(\frac{\pi x_{4,t-1}}{2}\Big)
+8\sin^2\!\Big(\frac{\pi x_{5,t-1}}{2}\Big)
+8\sin^2\!\Big(\frac{\pi x_{9,t-1}}{2}\Big)$\\

6 & $\valpha = 0.4$ &
$m(\mathbf{x})=
10\sin^2\!\Big(\frac{\pi x_{1,t-1}}{2}\Big)
\sin^2\!\Big(\frac{\pi x_{3,t-1}}{2}\Big)
+10\sin^2\!\Big(\frac{\pi x_{4,t-1}}{2}\Big)
\sin^2\!\Big(\frac{\pi x_{5,t-1}}{2}\Big)
+10\sin^2\!\Big(\frac{\pi x_{9,t-1}}{2}\Big)$\\

7 & $\valpha = 0.4$ &
$m(\mathbf{x})=
8x_{1,t-1}
+8x_{3,t-1}
+8x_{4,t-1}
+8x_{5,t-1}
+8x_{9,t-1}$\\

8 & $\valpha = 0.4$ &
$m(\mathbf{x})=
(x_{1,t-1}+x_{4,t-1})x_{3,t-1}
+x_{5,t-1}x_{9,t-1}$\\
\hline
\end{tabular}
}
\end{center}
\end{table}

Although the candidate covariate space is moderately sized, only a strict subset of 
variables truly contributes to the conditional variance dynamics. For example, in 
the five-covariate specifications, only covariates $x_{1,t-1}$, $x_{3,t-1}$, and 
$x_{4,t-1}$ are relevant to the volatility function, while the remaining covariates are purely 
noise. Similarly, in the ten-covariate designs, only $x_{1,t-1}$, $x_{3,t-1}$, 
$x_{4,t-1}$, $x_{5,t-1}$, and $x_{9,t-1}$ are relevant, with the other variables 
having no structural effect. This sparse but structured signal creates a setting 
in which irrelevant exogenous covariates may appear spuriously associated with volatility 
due to sampling variability, especially in finite samples. Moreover, the presence 
of nonlinear, interaction, and non-additive components poses further  challenges to
variable selection, as standard linear screening methods may fail to capture the true 
functional form of the dependence. 

In this section we also generated the covariates with a Multivariate Normal distribution with exponential decay covariance as in the previous section. After 200 Monte Carlo simulation runs we compute the mean number of correctly selected  covariates, the mean number of incorrectly selected covariates, and their conterparts number of correctly excluded and incorrectly excluded covariates. In addition, we compute the percentage of simulation runs that each covariate was selected. 

In Tables \ref{tab:5-covariate} and \ref{tab:10-covariate}, we report the average counts for correctly selected (C.S.), incorrectly selected (I.S.), correctly excluded (C.E.), and incorrectly excluded (I.E.) variables, computed over 200 independent iterations. For the 5-covariate scenario, which includes two relevant and three non-relevant covariates, the ideal values for C.S. and C.E. are 2 and 3, respectively. In the 10-covariate case, which includes five relevant and five non-relevant covariates, the ideal count for both C.S. and C.E. is 5. Tables \ref{tab:5-covariate-M1} and \ref{tab:10-covariate-M1} present percentage of simulation runs each covariate was selected in Model 1 with Normal shock distribution and $n=1000$ 
for the 5-covariate and 10-covariate models respectively.
The complete simulation results across all simulation settings can be found in the Supplementary Material.

The results in Tables~\ref{tab:5-covariate} and \ref{tab:10-covariate} show a clear contrast between the proposed procedure and its linear competitors when the conditional variance depends on nonlinear covariate effects. In the nonlinear additive and quadratic designs (Models 1--5), ARM generally attains substantially larger numbers of correctly selected variables together with very small numbers of incorrectly selected variables. This pattern is particularly pronounced in the 5-covariate setting under Normal innovations, where the proposed method selects the correct variables and excludes the irrelevant ones with high probability in Models 1--5, while the linear method of \cite{FrancqThieu2019} fails almost completely in the purely nonlinear specifications and LASSO tends to recover only part of the active set while selecting many irrelevant variables. These findings are consistent with the fact that the signal enters through nonlinear transformations, so methods based on linear screening are not well aligned with the true structure of the volatility function.

The advantage of the proposed procedure is especially visible in models containing interactions and other non-additive effects. In Models 6 and 8, where the relevant covariates act jointly rather than only through marginal linear contributions, the proposed method still recovers a substantially larger portion of the active set than the competing procedures in the Normal shock case, while keeping false selections close to zero. This indicates that the proposed method is able to detect variables that are important through their contribution to a nonlinear regression surface, even when their effect may be weak or entirely hidden from a linear main-effect representation. By contrast, the Linear and LASSO procedures can only succeed to the extent that the nonlinear signal induces an incidental linear association, which explains their weaker and less stable performance across the nonlinear designs.

The linear benchmark in Model~7 further clarifies this comparison. In that case, the simple linear procedure performs best, as expected, because the model exactly matches its assumption, whereas the proposed method is not designed specifically for purely linear effects.  The covariate-wise selection frequencies in Tables~\ref{tab:5-covariate-M1} and \ref{tab:10-covariate-M1} reinforce this conclusion. As the sample size increases, the proposed method increasingly concentrates selection on the truly active variables and largely excludes inactive ones, whereas LASSO selects relevant and irrelevant covariates with similar frequencies and the linear method misses the nonlinear signals altogether. The weaker performance in the $d=10$ setting reflects the greater difficulty of recovering the active set when the number of irrelevant covariates increases, since the signal becomes more diluted and finite-sample separation between relevant and non-relevant variables becomes harder. This effect is especially noticeable for nonlinear and interaction-based specifications, where a higher  dimension requires substantially more information and sample size for the method to accurately detect the variables entering the conditional variance function.

\begin{table}[htbp]
\centering
\footnotesize
\setlength{\tabcolsep}{5pt}
\caption{Simulation Results for $d=5, n=10000, \rho=0$.}\label{tab:5-covariate}
\begin{tabular}{ccc|rrr|rrr|rrr}
\hline
\multirow{1}{*}{} &
\multirow{3}{*}{\textbf{Model}} &
\multirow{3}{*}{\textbf{Metric}} &
\multicolumn{9}{c}{\textbf{Shock Distribution}} \\

\multicolumn{3}{c|}{} &
\multicolumn{3}{c}{\textbf{Normal}} &
\multicolumn{3}{c}{\textbf{T}} &
\multicolumn{3}{c}{\textbf{Laplace}} \\

\multicolumn{3}{c|}{}  &  
Linear & LASSO & ARM &
Linear & LASSO & ARM &
Linear & LASSO & ARM \\
\hline

& \multirow{4}{*}{1} 
& C.S. & 0.00 & 1.67 & 3.00 & 0.00 & 1.53 & 2.97 & 0.00 & 1.46 & 2.90 \\
& & I.S. & 0.08 & 1.03 & 0.06 & 0.01 & 1.01 & 0.04 & 0.01 & 0.99 & 0.04 \\
& & C.E. & 1.92 & 0.97 & 1.95 & 1.99 & 0.98 & 1.96 & 1.99 & 1.00 & 1.97 \\
& & I.E. & 3.00 & 1.33 & 0.00 & 3.00 & 1.47 & 0.03 & 3.00 & 1.54 & 0.10 \\
\cmidrule(lr){2-12}

& \multirow{4}{*}{2} 
& C.S. & 0.00 & 1.45 & 2.98 & 0.00 & 1.50 & 2.26 & 0.00 & 1.44 & 1.82 \\
& & I.S. & 0.02 & 1.24 & 0.02 & 0.01 & 1.15 & 0.04 & 0.01 & 1.16 & 0.04 \\
& & C.E. & 1.98 & 0.76 & 1.98 & 1.99 & 0.85 & 1.96 & 2.00 & 0.84 & 1.97 \\
& & I.E. & 3.00 & 1.55 & 0.01 & 3.00 & 1.50 & 0.74 & 3.00 & 1.56 & 1.18 \\
\cmidrule(lr){2-12}

& \multirow{4}{*}{3} 
& C.S. & 0.00 & 1.50 & 3.00 & 0.00 & 1.51 & 2.98 & 0.00 & 1.49 & 2.88 \\
& & I.S. & 0.12 & 1.14 & 0.06 & 0.06 & 1.06 & 0.04 & 0.04 & 0.99 & 0.04 \\
& & C.E. & 1.88 & 0.86 & 1.95 & 1.94 & 0.94 & 1.97 & 1.96 & 1.00 & 1.97 \\
& & I.E. & 3.00 & 1.50 & 0.00 & 3.00 & 1.49 & 0.03 & 3.00 & 1.51 & 0.12 \\
\cmidrule(lr){2-12}

& \multirow{4}{*}{4} 
& C.S. & 0.00 & 1.43 & 2.94 & 0.00 & 1.47 & 2.23 & 0.00 & 1.44 & 1.67 \\
& & I.S. & 0.15 & 1.15 & 0.04 & 0.06 & 1.09 & 0.04 & 0.04 & 1.23 & 0.01 \\
& & C.E. & 1.85 & 0.85 & 1.96 & 1.94 & 0.91 & 1.97 & 1.97 & 0.78 & 2.00 \\
& & I.E. & 3.00 & 1.57 & 0.06 & 3.00 & 1.53 & 0.77 & 3.00 & 1.56 & 1.33 \\
\cmidrule(lr){2-12}

& \multirow{4}{*}{5} 
& C.S. & 0.25 & 1.50 & 2.99 & 0.69 & 1.28 & 2.21 & 0.21 & 1.23 & 1.51 \\
& & I.S. & 0.09 & 0.81 & 0.04 & 0.29 & 0.86 & 0.05 & 0.09 & 0.85 & 0.03 \\
& & C.E. & 1.91 & 1.19 & 1.97 & 1.71 & 1.14 & 1.95 & 1.92 & 1.15 & 1.98 \\
& & I.E. & 2.75 & 1.50 & 0.01 & 2.31 & 1.72 & 0.79 & 2.79 & 1.77 & 1.49 \\
\cmidrule(lr){2-12}

& \multirow{4}{*}{6} 
& C.S. & 0.15 & 1.58 & 2.92 & 0.12 & 1.39 & 2.15 & 0.34 & 1.44 & 1.77 \\
& & I.S. & 0.06 & 0.97 & 0.04 & 0.02 & 0.91 & 0.06 & 0.14 & 0.78 & 0.03 \\
& & C.E. & 1.94 & 1.03 & 1.97 & 1.98 & 1.09 & 1.94 & 1.86 & 1.22 & 1.97 \\
& & I.E. & 2.85 & 1.42 & 0.07 & 2.88 & 1.61 & 0.85 & 2.65 & 1.56 & 1.23 \\
\cmidrule(lr){2-12}

& \multirow{4}{*}{7} 
& C.S. & 3.00 & 1.83 & 1.19 & 3.00 & 2.30 & 0.33 & 3.00 & 2.40 & 0.17 \\
& & I.S. & 0.01 & 0.57 & 0.03 & 0.01 & 0.77 & 0.01 & 0.01 & 0.74 & 0.01 \\
& & C.E. & 2.00 & 1.43 & 1.98 & 1.99 & 1.23 & 2.00 & 2.00 & 1.26 & 1.99 \\
& & I.E. & 0.00 & 1.17 & 1.81 & 0.01 & 0.70 & 2.67 & 0.00 & 0.59 & 2.83 \\
\cmidrule(lr){2-12}

& \multirow{4}{*}{8} 
& C.S. & 1.97 & 1.99 & 2.69 & 1.75 & 1.88 & 1.68 & 1.60 & 1.85 & 1.42 \\
& & I.S. & 0.00 & 0.49 & 0.04 & 0.00 & 0.58 & 0.04 & 0.00 & 0.52 & 0.03 \\
& & C.E. & 2.00 & 1.50 & 1.97 & 2.00 & 1.42 & 1.96 & 2.00 & 1.48 & 1.97 \\
& & I.E. & 1.03 & 1.01 & 0.30 & 1.25 & 1.12 & 1.32 & 1.40 & 1.15 & 1.58 \\
\hline
\end{tabular}
\end{table}

\begin{table}[htbp]
\centering
\footnotesize
\setlength{\tabcolsep}{5pt}
\caption{Simulation Results for $d=10, n=10000, \rho=0$.}\label{tab:10-covariate}
\begin{tabular}{ccc|rrr|rrr|rrr}
\hline
\multirow{1}{*}{} &
\multirow{3}{*}{\textbf{Model}} &
\multirow{3}{*}{\textbf{Metric}} &
\multicolumn{9}{c}{\textbf{Shock Distribution}} \\

\multicolumn{3}{c|}{} &
\multicolumn{3}{c}{\textbf{Normal}} &
\multicolumn{3}{c}{\textbf{T}} &
\multicolumn{3}{c}{\textbf{Laplace}} \\

\multicolumn{3}{c|}{}&  
Linear & LASSO & ARM &
Linear & LASSO & ARM &
Linear & LASSO & ARM \\
\hline

& \multirow{4}{*}{1} & C.S. & 0.03 & 2.65 & 3.92 & 0.00 & 2.44 & 1.51 & 0.00 & 2.48 & 0.76 \\
& & I.S. & 0.92 & 2.65 & 0.04 & 0.16 & 2.55 & 0.03 & 0.12 & 2.46 & 0.01 \\
& & C.E. & 4.08 & 2.35 & 4.96 & 4.84 & 2.45 & 4.97 & 4.88 & 2.54 & 4.99 \\
& & I.E. & 4.97 & 2.35 & 1.08 & 5.00 & 2.56 & 3.49 & 5.00 & 2.52 & 4.25 \\
\cmidrule(lr){2-12}
& \multirow{4}{*}{2} & C.S. & 0.00 & 2.40 & 0.66 & 0.00 & 2.54 & 0.19 & 0.00 & 2.59 & 0.19 \\
& & I.S. & 0.34 & 2.59 & 0.04 & 0.08 & 2.58 & 0.01 & 0.07 & 2.60 & 0.04 \\
& & C.E. & 4.67 & 2.41 & 4.96 & 4.92 & 2.42 & 5.00 & 4.92 & 2.40 & 4.96 \\
& & I.E. & 5.00 & 2.60 & 4.34 & 5.00 & 2.46 & 4.81 & 5.00 & 2.41 & 4.81 \\
\cmidrule(lr){2-12}
& \multirow{4}{*}{3} & C.S. & 0.00 & 2.63 & 3.73 & 0.00 & 2.45 & 1.31 & 0.00 & 2.62 & 0.74 \\
& & I.S. & 0.83 & 2.56 & 0.06 & 0.14 & 2.71 & 0.03 & 0.07 & 2.62 & 0.01 \\
& & C.E. & 4.17 & 2.44 & 4.94 & 4.86 & 2.29 & 4.97 & 4.92 & 2.38 & 4.99 \\
& & I.E. & 5.00 & 2.37 & 1.27 & 5.00 & 2.55 & 3.69 & 5.00 & 2.38 & 4.25 \\
\cmidrule(lr){2-12}
& \multirow{4}{*}{4} & C.S. & 4.81 & 3.90 & 4.29 & 4.15 & 4.25 & 2.31 & 3.88 & 4.44 & 1.47 \\
& & I.S. & 0.00 & 1.73 & 0.06 & 0.00 & 2.00 & 0.06 & 0.00 & 2.23 & 0.06 \\
& & C.E. & 5.00 & 3.27 & 4.94 & 5.00 & 3.00 & 4.95 & 5.00 & 2.77 & 4.94 \\
& & I.E. & 0.19 & 1.10 & 0.71 & 0.85 & 0.76 & 2.69 & 1.11 & 0.56 & 3.54 \\
\cmidrule(lr){2-12}
& \multirow{4}{*}{5} & C.S. & 0.85 & 2.54 & 0.53 & 0.43 & 2.27 & 0.12 & 0.12 & 2.27 & 0.06 \\
& & I.S. & 0.64 & 2.55 & 0.01 & 0.27 & 2.23 & 0.02 & 0.10 & 2.38 & 0.01 \\
& & C.E. & 4.36 & 2.45 & 4.99 & 4.74 & 2.77 & 4.98 & 4.90 & 2.62 & 5.00 \\
& & I.E. & 4.15 & 2.46 & 4.47 & 4.57 & 2.73 & 4.88 & 4.88 & 2.73 & 4.94 \\
\cmidrule(lr){2-12}
& \multirow{4}{*}{6} & C.S. & 1.44 & 2.64 & 1.95 & 0.23 & 2.58 & 0.99 & 0.11 & 2.64 & 0.74  \\
& & I.S. & 0.90 & 2.66 & 0.04 & 0.15 & 2.56 & 0.05 & 0.09 & 2.53 & 0.05 \\
& & C.E. & 4.10 & 2.34 & 4.96 & 4.85 & 2.44 & 4.95 & 4.91 & 2.48 & 4.96 \\
& & I.E. & 3.56 & 2.36 & 3.05 & 4.78 & 2.42 & 4.01 & 4.89 & 2.37 & 4.26 \\
\cmidrule(lr){2-12}
& \multirow{4}{*}{7} & C.S. & 5.00 & 3.10 & 0.07 & 4.71 & 3.83 & 0.04 &4.20 & 3.97 & 0.01 \\
& & I.S. & 0.03 & 1.92 & 0.02 & 0.02 & 2.25 & 0.01 &0.01 & 2.22 & 0.03\\
& & C.E. & 4.97 & 3.08 & 4.98 & 4.98 & 2.75 & 5.00 & 4.99 & 2.79 & 4.97 \\
& & I.E. & 0.01 & 1.90 & 4.92 & 0.29 & 1.17 & 4.96 & 0.80 & 1.04 & 4.99\\
\cmidrule(lr){2-12}
& \multirow{4}{*}{8} & C.S. & 4.50 & 3.12 & 1.14 & 3.36 & 3.04 & 0.40 & 2.98 & 2.93 & 0.25 \\
& & I.S. & 0.00 & 2.22 & 0.03 & 0.00 & 2.02 & 0.02 & 0.00 & 2.25 & 0.02 \\
& & C.E. & 5.00 & 2.78 & 4.97 & 5.00 & 2.98 & 4.98 & 5.00 & 2.75 & 4.98\\
& & I.E. & 0.50 & 1.88 & 3.86 & 1.64 & 1.96 & 4.61 & 2.02 & 2.08 & 4.75 \\
\hline
\end{tabular}
\end{table}

\begin{table}[http!]
\centering
\caption{Proportion of covariate selection out of 500 iterations with normal shock distribution based on the 5-covariate Model 1. The active covariates are $X_1, X_3$ and $X_4$.}
\label{tab:5-covariate-M1}
\begin{tabular}{cclccccc}
\toprule
$\rho$ & $n$ & Method & $X_1$ & $X_2$ & $X_3$ & $X_4$ & $X_5$ \\
\midrule
\multirow{12}{*}{0.0} & \multirow{3}{*}{500} & Linear & 0.000 & 0.010 & 0.000 & 0.000 & 0.020 \\
 & & Lasso & 0.580 & 0.565 & 0.615 & 0.575 & 0.590 \\
 & & ARM & 0.180 & 0.010 & 0.120 & 0.170 & 0.025 \\
\cmidrule{2-8}
 & \multirow{3}{*}{1000} & Linear & 0.000 & 0.045 & 0.000 & 0.000 & 0.015 \\
 & & Lasso & 0.500 & 0.540 & 0.465 & 0.545 & 0.615 \\
 & & ARM & 0.410 & 0.025 & 0.375 & 0.355 & 0.010 \\
\cmidrule{2-8}
 & \multirow{3}{*}{5000} & Linear & 0.000 & 0.015 & 0.000 & 0.000 & 0.020 \\
 & & Lasso & 0.535 & 0.555 & 0.490 & 0.525 & 0.470 \\
 & & ARM & 0.995 & 0.025 & 1.000 & 0.990 & 0.010 \\
\cmidrule{2-8}
 & \multirow{3}{*}{10000} & Linear & 0.000 & 0.040 & 0.000 & 0.000 & 0.040 \\
 & & Lasso & 0.570 & 0.515 & 0.575 & 0.525 & 0.515 \\
 & & ARM & 1.000 & 0.030 & 1.000 & 1.000 & 0.025 \\
\midrule
\multirow{12}{*}{0.5} & \multirow{3}{*}{500} & Linear & 0.000 & 0.000 & 0.000 & 0.000 & 0.005 \\
 & & Lasso & 0.595 & 0.655 & 0.610 & 0.725 & 0.660 \\
 & & ARM & 0.215 & 0.040 & 0.305 & 0.360 & 0.035 \\
\cmidrule{2-8}
 & \multirow{3}{*}{1000} & Linear & 0.000 & 0.000 & 0.000 & 0.000 & 0.000 \\
 & & Lasso & 0.560 & 0.580 & 0.510 & 0.550 & 0.545 \\
 & & ARM & 0.395 & 0.065 & 0.660 & 0.640 & 0.010 \\
\cmidrule{2-8}
 & \multirow{3}{*}{5000} & Linear & 0.000 & 0.000 & 0.000 & 0.000 & 0.000 \\
 & & Lasso & 0.485 & 0.480 & 0.520 & 0.510 & 0.540 \\
 & & ARM & 1.000 & 0.360 & 1.000 & 1.000 & 0.065 \\
\cmidrule{2-8}
 & \multirow{3}{*}{10000} & Linear & 0.000 & 0.000 & 0.000 & 0.000 & 0.010 \\
 & & Lasso & 0.565 & 0.505 & 0.525 & 0.500 & 0.510 \\
 & & ARM & 1.000 & 0.705 & 1.000 & 1.000 & 0.100 \\
\bottomrule
\end{tabular}
\end{table}

\begin{table}[http!]
\centering
\caption{Proportion of covariate selection out of 500 iterations with normal shock distribution based on the 10-covariate Model 1. The active covariates are $X_1, X_3, X_4, X_5,\text{ and }X_9$.}\label{tab:10-covariate-M1}
\resizebox{\textwidth}{!}{
\begin{tabular}{cclcccccccccc}
\toprule
$\rho$ & $n$ & Method & $X_1$ & $X_2$ & $X_3$ & $X_4$ & $X_5$ & $X_6$ & $X_7$ & $X_8$ & $X_9$ & $X_{10}$ \\
\midrule
\multirow{12}{*}{0.0} & \multirow{3}{*}{500} & Linear & 0.000 & 0.010 & 0.000 & 0.000 & 0.000 & 0.020 & 0.010 & 0.010 & 0.005 & 0.010 \\
  &   & Lasso & 0.630 & 0.620 & 0.635 & 0.650 & 0.695 & 0.670 & 0.605 & 0.655 & 0.665 & 0.665 \\
  &   & ARM & 0.015 & 0.010 & 0.015 & 0.005 & 0.020 & 0.000 & 0.005 & 0.010 & 0.005 & 0.000 \\
\cmidrule{2-13}
  & \multirow{3}{*}{1000} & Linear & 0.000 & 0.005 & 0.000 & 0.000 & 0.000 & 0.000 & 0.005 & 0.020 & 0.000 & 0.015 \\
  &   & Lasso & 0.540 & 0.640 & 0.575 & 0.575 & 0.535 & 0.605 & 0.625 & 0.550 & 0.560 & 0.550 \\
  &   & ARM & 0.020 & 0.005 & 0.060 & 0.040 & 0.030 & 0.005 & 0.005 & 0.010 & 0.050 & 0.015 \\
\cmidrule{2-13}
  & \multirow{3}{*}{5000} & Linear & 0.000 & 0.045 & 0.000 & 0.000 & 0.000 & 0.040 & 0.050 & 0.060 & 0.000 & 0.040 \\
  &   & Lasso & 0.490 & 0.560 & 0.595 & 0.475 & 0.535 & 0.495 & 0.500 & 0.530 & 0.510 & 0.555 \\
  &   & ARM & 0.420 & 0.015 & 0.325 & 0.305 & 0.390 & 0.000 & 0.005 & 0.005 & 0.450 & 0.015 \\
\cmidrule{2-13}
  & \multirow{3}{*}{10000} & Linear & 0.015 & 0.190 & 0.000 & 0.005 & 0.010 & 0.175 & 0.200 & 0.150 & 0.000 & 0.205 \\
  &   & Lasso & 0.525 & 0.580 & 0.520 & 0.530 & 0.555 & 0.520 & 0.525 & 0.485 & 0.525 & 0.535 \\
  &   & ARM & 0.795 & 0.015 & 0.810 & 0.745 & 0.780 & 0.010 & 0.020 & 0.000 & 0.785 & 0.000 \\
\midrule
\multirow{12}{*}{0.5} & \multirow{3}{*}{500} & Linear & 0.000 & 0.000 & 0.000 & 0.000 & 0.000 & 0.000 & 0.000 & 0.000 & 0.000 & 0.000 \\
  &   & Lasso & 0.690 & 0.690 & 0.620 & 0.685 & 0.580 & 0.620 & 0.680 & 0.705 & 0.640 & 0.620 \\
  &   & ARM & 0.035 & 0.015 & 0.055 & 0.055 & 0.050 & 0.005 & 0.000 & 0.000 & 0.040 & 0.005 \\
\cmidrule{2-13}
  & \multirow{3}{*}{1000} & Linear & 0.000 & 0.000 & 0.000 & 0.000 & 0.000 & 0.000 & 0.000 & 0.000 & 0.000 & 0.000 \\
  &   & Lasso & 0.560 & 0.575 & 0.600 & 0.665 & 0.620 & 0.600 & 0.605 & 0.540 & 0.540 & 0.535 \\
  &   & ARM & 0.020 & 0.005 & 0.095 & 0.170 & 0.125 & 0.000 & 0.010 & 0.000 & 0.020 & 0.015 \\
\cmidrule{2-13}
  & \multirow{3}{*}{5000} & Linear & 0.000 & 0.000 & 0.000 & 0.000 & 0.000 & 0.000 & 0.000 & 0.000 & 0.000 & 0.030 \\
  &   & Lasso & 0.495 & 0.560 & 0.460 & 0.535 & 0.470 & 0.485 & 0.515 & 0.555 & 0.470 & 0.505 \\
  &   & ARM & 0.445 & 0.065 & 0.785 & 0.935 & 0.765 & 0.035 & 0.010 & 0.020 & 0.350 & 0.020 \\
\cmidrule{2-13}
  & \multirow{3}{*}{10000} & Linear & 0.000 & 0.015 & 0.000 & 0.000 & 0.000 & 0.005 & 0.015 & 0.005 & 0.000 & 0.100 \\
  &   & Lasso & 0.500 & 0.435 & 0.450 & 0.570 & 0.435 & 0.475 & 0.525 & 0.475 & 0.460 & 0.565 \\
  &   & ARM & 0.890 & 0.135 & 0.980 & 0.975 & 0.970 & 0.050 & 0.010 & 0.015 & 0.810 & 0.030 \\
\bottomrule
\end{tabular}
}
\end{table}

\section{Real Data Analysis} \label{sec.real_data}



Modeling the volatility of financial returns is a central problem in financial econometrics. 
Equity index returns, such as those of the S\&P 500, are well known to exhibit volatility clustering, 
heavy tails, and time-varying conditional heteroskedasticity. 
ARCH-type models provide a natural framework for capturing such features by allowing the conditional variance to evolve dynamically as a function of past shocks.

However, financial markets do not evolve in isolation. 
Macroeconomic conditions, global equity markets, commodity prices, and currency fluctuations 
can all influence the level and persistence of volatility in U.S. equity markets. 
In particular, the S\&P 500 is beleived to be linked to global market movements, 
energy and metal prices, and exchange rate dynamics due to international trade, capital flows, 
and supply chain interdependence. 
Ignoring these external influences may lead to misspecified volatility dynamics 
and omitted-variable bias in the conditional variance equation.

To address this issue, we fit the proposed ARCH-$m(X)$ with $p=1$ and
nonparametric covariate effects in the variance equation. 
This framework allows the conditional variance of S\&P 500 log returns 
to depend flexibly on a set of exogenous covariates, without imposing restrictive linear functional forms. 
The nonparametric component is particularly appealing because 
the relationship between volatility and global or commodity factors 
may be nonlinear.

The primary time series under consideration is the log return of the S\&P 500 index, defined as
\[
r_t = \log(P_t) - \log(P_{t-1}),
\]
where $P_t$ denotes the S\&P 500 closing price at time $t$. 
Similarly, for each covariate we consider log returns to ensure stationarity 
and comparability across series.
The set of candidate covariates includes: 
     China Market (SSE Composite Index - Shanghai Stock Exchange),
     Asia Market (iShares MSCI All Country Asia ex Japan ETF),
     Crude Oil (West Texas Intermediate (WTI) Crude Oil Futures),
     USD Index (ICE U.S. Dollar Index),
     Gold (Gold Continuous Futures - COMEX),
     Copper (Copper Continuous Futures - COMEX),
     Silver (Silver Continuous Futures - COMEX),
     Steel (SLX - VanEck Steel ETF),
     Rice (CBOT Rough Rice Futures),
     Wheat (SLX - VanEck Steel ETF), and
     Europe Market (EURO STOXX 50 Index).

These variables represent major global equity markets (China, Asia, Europe), 
energy markets (Crude Oil), currency markets (USD Index), 
precious and industrial metals (Gold, Copper, Silver, Steel Proxy), 
and agricultural commodities (Rice, Wheat). 
Together, they provide a broad representation of global macro-financial conditions.

Although incorporating many covariates increases modeling flexibility, 
including irrelevant variables can inflate estimation variance and 
reduce interpretability.
In high-dimensional volatility modeling, variable selection is therefore crucial. 
It ensures parsimony, improves statistical efficiency, 
and helps identify economically meaningful drivers of market volatility.

To determine which covariates significantly contribute to the conditional variance, 
we compute adjusted $p$-values using the proposed nonparametric test-based variable selection with False Discovery Rate (FDR) corrections 
to control for multiple testing. 
The adjusted $p$-values are:
$\approx$ 1 for China Market,  
$\approx$ 0 for Asia Market 
$\approx$ 0.00009 for Crude Oil
$\approx$ 1 for USD Index
$\approx$ 0.002 for Gold
$\approx$ 1 for Copper
$\approx$ 0.082 for Silver 
$\approx$ 0 for Steel
$\approx$ 1 for Rice
$\approx$ 0.061 for Wheat, and
$\approx$ 0.0000002 for Europe Market.
Using a significance threshold of $0.05$ after FDR correction, 
we selected the following five covariates: Asia Market, Crude Oil, Gold, Steel, and Europe Market.

The selected variables provide economically intuitive insights. 
For example, the significance of both the Asia Market and the Europe Market highlight the strong interdependence between U.S. equity volatility 
and other major global equity markets. This finding reflects the globalization of financial markets and the rapid transmission of shocks across regions. In addition, Crude Oil is also significant, consistent with the role of energy prices as a key driver of macroeconomic uncertainty and corporate profitability. It is well known that fluctuations in oil prices can directly affect inflation expectations, production costs, and sectoral performance, which in turn influence aggregate market volatility. Moreover, Gold appears as a significant covariate, which is particularly noteworthy. Gold is often viewed as a safe-haven asset; its price movements may reflect shifts in risk sentiment. A significant nonparametric effect of Gold on S\&P 500 volatility suggests that changes in investor risk aversion are closely tied to equity market uncertainty. Furthermore, Steel is also selected, indicating that industrial commodity dynamics may contain information about economic activity and global demand conditions. 
Its significance suggests that real-sector fundamentals are embedded in volatility dynamics.

In contrast, the China Market, the USD Index, Copper, Rice, and Wheat 
are not selected at the $0.05$ level after FDR correction, 
while Silver and Wheat show borderline significance 
(with adjusted $p$-values slightly above $0.05$). 
This indicates that, conditional on the selected variables, 
their marginal contribution to explaining S\&P 500 volatility 
is limited in this specification.

Overall, the results demonstrate that only a subset of global and commodity variables 
exerts a statistically significant influence on S\&P 500 volatility. 
The ARCH-$m(X)$ framework combined with FDR-based variable selection 
provides a principled approach to isolating these key drivers, 
leading to a more parsimonious and interpretable volatility model.

\section{Appendix}
\subsection{Proof of Theorem \ref{thm.null}}
\begin{proof}

The proof of Theorem \ref{thm.null} is given for the ARCH(1)-m(X) case. The proof for ARCH(p)-m(X) can be shown in a similar way. In this proof we follow steps similar to those in \cite{ZambomGel2020}.

Given the estimators $\hat{\alpha}_1$ and $\hat{m}(\cdot)$ as in equation (\ref{eq.v_hat},) write the null hypothesis residuals as
\begin{eqnarray}
\hat{v}_{t} &=&  \epsilon_t^2  - \hat{\alpha}_1\epsilon_{t-1}^2 - \hat{m}_\ell(\vX^{-\ell}_{t-1})\nonumber\\
&=&  \epsilon_t^2 - \alpha_1\epsilon_{t-1}^2 - m_\ell(\vX^{-\ell}_{t-1})  - (\hat{\alpha}_1 - \alpha_1)\epsilon_{t-1}^2 - (\hat{m}_\ell(\vX^{-\ell}_{t-1}) - m_\ell(\vX^{-\ell}_{t-1}))   \nonumber\\
&=&v_t -  (\hat{\alpha}_1 - \alpha_1)\epsilon_{t-1}^2 - ({\hat{\vtheta}^{-\ell}}-\vtheta^{-\ell})^T \mathcal{B}(\vX_{t-1}^{-\ell}).
\end{eqnarray}

Let $\vV = (\hat{v}_t, t \in W_1, \ldots, \hat{v}_t, t \in W_n)$, $\valpha_V = \Big((\hat{\alpha}_1 - \alpha_1)\epsilon_{t-1}^2, t \in W_1, \ldots,$ $(\hat{\alpha}_1 - \alpha_1)\epsilon_{t-1}^2, t\in W_n\Big)$, $\vv_V = (v_t,t \in W_1, \ldots, v_t, t \in W_n)$, and  $\vm_V = (({\hat{\vtheta}}^{-\ell}-\vtheta^{-\ell})^T \mathcal{B}(\vX_{t-1}^{-\ell}), t\in W_1, \ldots, ({\hat{\vtheta}}^{-\ell}-\vtheta^{-\ell})^T \mathcal{B}(\vX_{t-1}^{-\ell}), t\in W_n)$. The statistic $T_n = MST_n - MSE_n$ can be written as 
\bqan \label{T-stat}
T = \vV'\vA\vV =\vv_V'\vA\vv_V +  \valpha_V'\vA\valpha_V + \vm_V'\vA\vm_V  + 2\vv_V'\vA\valpha_V + 2\vv_V'\vA\vm_V  + 2\valpha_V'\vA\vm_V, 
\eqan
where 
\bqan\label{def.matrixA}
\vA = \frac{nk_n-1}{n(n-1)k_n(k_n-1)}\oplus_{i=1}^{n} J_{k_n}-\frac{1}{n(n-1)k_n} J_{nk_n}-\frac{1}{n(k_n-1)} I_{nk_n},
\eqan
$I_r$ is a identity matrix of dimension $r$, $J_r$ is a $r$x$r$ matrix of 1's and $\oplus$ is the Kronecker sum or direct sum.

In what follows, using steps similar to those in \citet{WangEtAl2010}, we show that asymptotic distribution of $n^{1/2}k_n^{-1/2}\vv_V'\vA\vv_V$ converges to a normal distribution and all other terms are $o_p(n^{-1/2}k_n^{1/2})$. First, note that
\bqa
\vv_V'\vA\vv_V = \frac{1}{n(k_n - 1)}\sum_{t=1}^n\sum_{t_1\neq t_2}^n v_{t_1}v_{t_2}I(t_1, t_2 \in W_t),
\eqa
where $v_t$ is white noise and $E(\vv_V'\vA\vv_V) = 0$. Hence, the variance of $\vv_V'\vA\vv_V$ can be computed as
\bqa
&&E(\vv_V'\vA\vv_V)^2 \\
&=& \frac{1}{n^2(k_n - 1)^2}\sum_{i_1=1}^n\sum_{i_2=1}^n\sum_{t_1\neq s_1}^n\sum_{t_2\neq s_2}^n E(v_{t_1}v_{s_1}v_{t_2}v_{s_2})I(t_1, s_1 \in W_{i_1}, t_2, s_2 \in W_{i_2})\\
&=& \frac{1}{n^2(k_n - 1)^2}\sum_{i_1=1}^n\sum_{i_2=1}^n\sum_{t \neq s}^n \tau^2I(t, s \in W_{i_1}\cap W_{i_2})\\
&& \to \begin{cases} 4\tau^2/3 & \text{ if } k_n \to \infty.\\
 \frac{2(2k_n -1)}{3(k_n-1)}\tau^2 & \text{ if } k_n \text{ is fixed}.
 \end{cases}
\eqa

To prove the asymptotic normality of $\vv_V'\vA\vv_V$, we first express $\vv_V'\vA\vv_V$ in the form of 
$\vv_V'\vA\vv_V = \frac{1}{n}\sum_{t=1}^nA_t = \frac{1}{n}S_n$, where $A_t = (1/(k_n-1))\sum_{t_1\neq t_2}v_{t_1}v_{t_2}I(t_1, t_2 \in W_t)$, and define
\bqa
U_{nt} &=& A_{(t-1)(b_n+l_n)+1} + \ldots + A_{(t-1)(b_n+l_n)+b_n}\\
V_{nt} &=& A_{(t-1)(b_n+l_n)+b_n+1} + \ldots + A_{t(b_n+l_n)},
\eqa
$t = 1, \ldots, r_n$, where
$b_n \sim n^{2/3}k_n^{1/3}, l_n \sim k_n, r_n \sim n/b_n = n^{1/3}k_n^{-1/3},$
so that
\bqa
S_n = \sum_{t=1}^{r_n}U_{nt} + \sum_{t=1}^{r_n}V_{nt}.
\eqa
It can be shown that $\sum_{t=1}^{r_n}V_{nt} = o_p((nk_n)^{-1/2})$. Since $U_{nt}$ are uncorrelated, Normality  $U_{nt}$ is established using  the $\rho$-mixing theorem
in \citep{Peligrad1987}.

In what follows, we will show that the remainder of all the terms on the right hand side of equation \eqref{T-stat} converge in probability to zero.

For the second term in equation \eqref{T-stat}, note that
\begin{eqnarray}\label{alpha_alpha}
\sqrt{\frac{n}{k_n}}\valpha_V'\vA\valpha_V &=& \sqrt{\frac{n}{k_n}}\frac{(nk_n-1)}{n(n-1)k_n(k_n-1)}(\hat{\alpha}_1 - \alpha_1)^2\sum_{t=1}^{n}\left(\sum_{\ell\in W_t}\epsilon_{\ell-1}^2\right)^2 
\\
&-&\sqrt{\frac{n}{k_n}}\frac{k_n}{n(n-1)}(\hat{\alpha}_1 - \alpha_1)^2\left(\sum_{t=1}^{n}\epsilon_{t-1}^2\right)^2 - \sqrt{\frac{n}{k_n}}\frac{k_n}{n(k_n-1)}(\hat{\alpha}_1 - \alpha_1)^2\sum_{t=1}^{n}\epsilon_{t-1}^4.\nonumber
\end{eqnarray}

In view of an ARCH model as defined in equation \eqref{eq:ARCHmodel} and given that $\epsilon_t^2$ can be represented as an autoregressive process of order $1$ as in equation \eqref{eq:model-new}, we find that under the null hypothesis $E(\epsilon_t^2)=m_i(\vX^{-i}_{t-1})/(1-\alpha_1)$ (the general case when $p > 1$ can be shown in a similar way to be $E(\epsilon_t^2)=m_i(\vX^{-i}_{t-1})/(1-\sum_{j=1}^p\alpha_j)$) and $E(\epsilon_t^4) = E(z_t^4)E(\epsilon_t^2)$ (see e.g. \cite{GiraitisEtAl2000}). Hence, using the fact that $(\hat{\alpha}_1 - \alpha_1) = O_p(n^{-1/2})$ (\cite{FrancqThieu2019}), the third therm on the right hand side of equation (\ref{alpha_alpha}) is $O_p(n^{-1/2}k_n^{-1/2}n^{-1}n) = o_p(1)$.
Additionally, the second term on the right hand side of (\ref{alpha_alpha}) is $O_p(n^{-3/2}k_n^{1/2}n^{-1}n^2) = O_p(n^{-1/2}k_n^{1/2}) = o_p(1)$. Finally, the first term on the right hand side of (\ref{alpha_alpha}) is $O_p(n^{-1/2}k_n^{-3/2}n^{-1}nk_n^2) = O_p(n^{-1/2}k_n^{1/2}) = o_p(1)$.

For the fourth term on the right-hand-side of equation \eqref{T-stat}, note that
\bqan  \label{v_alpha}
 \sqrt{\frac{n}{k_n}}\vv_V'\vA\valpha_V&=& \sqrt{\frac{n}{k_n}}\frac{(nk_n -1)}{n(n-1)k_n(k_n-1)}(\hat{\alpha}_1 - \alpha_1)\sum_{t=1}^{n}\sum_{s \in W_t}v_s\sum_{s\in W_t}\epsilon^2_{s-1} 
\\
 &-&\sqrt{\frac{n}{k_n}}\frac{k_n}{n(n-1)}(\hat{\alpha}_1 - \alpha_1)\sum_{t=1}^{n}v_t\sum_{s=1}^{n}\epsilon^2_{s-1} -\sqrt{\frac{n}{k_n}} \frac{k_n}{n(k_n-1)}(\hat{\alpha}_1 - \alpha_1)\sum_{t=1}^{n}v_t\epsilon^2_{t-1}. \nonumber
\eqan
 Using the fact that under the null hypothesis (and $p = 1$)
$\epsilon_t^2 =  \alpha_1\epsilon_{t-1}^2 + m_i(\vX^{-\ell}_{t-1}) + v_t$ and $v_t$ is white noise, it is easy to see that $E(v_t\epsilon^2_{t-1}) = 0$, and $E(\sum_{t=1}^nv_t\epsilon^2_{t-1})^2 = \sum_{t=1}^nE(v_t^2\epsilon^4_{t-1}) + 2\sum_{t < s}E(v_t\epsilon^2_{t-1}v_s\epsilon^2_{s-1}) = n(\tau m_i(\vX_{t-2}^{-\ell}) + \tau^2 + 2\alpha_1 m_\ell(\vX_{t-2}^{-\ell})\tau m_\ell(\vX_{t-3}^{-\ell})/(1-\alpha_1)) + 0 = O(n)$. Hence, the last term on  the right hand side of (\ref{v_alpha}) is of order $O_p(n^{-1/2}k_n^{-1/2}n^{-1/2}n^{1/2}) = o_p(1)$.
By the Marcinkiewicz-Zygmund inequality for weakly dependent processes (\cite{DedeckerEtAl2007} - Theorem 4.1), $\sum_{t=1}^n\epsilon^2_{t-1} = O_p(n^{1/2})$. In addition, by the Martingale Central Limit Theorem, $\sum_{t=1}^nv_t = O_p(n^{1/2})$. Consequently, the second term on the right hand side of equation (\ref{v_alpha}) is of order $O_p(n^{-3/2}k_n^{1/2}n^{-1/2}n^{1/2}n^{1/2}) = o_p(1)$. Similarly, the first term on the right hand side of (\ref{v_alpha}) is of order $O_p(n^{-1/2}k_n^{-3/2}n^{-1/2}nk_n^{1/2}k_n^{1/2}) = o_p(1)$.

Now, using the fact that the B-Splines basis functions are bounded by some constant $M$, we re-express the fifth term in equation \eqref{T-stat} in the following form
\bqan  \label{v_m}
\sqrt{\frac{n}{k_n}}\vv_V'\vA\vm_V 
 &=&  \sqrt{\frac{n}{k_n}}\frac{(nk_n -1)}{n(n-1)k_n(k_n-1)}({\hat{\vtheta}}^{-\ell}-\vtheta^{-\ell})^T\sum_{t=1}^{n}\sum_{k \in W_t}v_k\sum_{s\in W_t} \mathcal{B}(\vX_{s-1}^{-\ell})\nonumber\\
&& \hspace{-.5cm} - \sqrt{\frac{n}{k_n}}\frac{k_n}{n(n-1)}({\hat{\vtheta}}^{-\ell}-\vtheta^{-\ell})^T\sum_{t=1}^{n}v_t\sum_{s=1}^{n} \mathcal{B}(\vX_{s-1}^{-\ell})  \nonumber\\
&& -\sqrt{\frac{n}{k_n}} \frac{k_n}{n(k_n-1)}({\hat{\vtheta}}^{-\ell}-\vtheta^{-\ell})^T\sum_{t=1}^{n}v_t \mathcal{B}(\vX_{t-1}^{-\ell}) \nonumber\\
&=&  \sqrt{\frac{n}{k_n}}\frac{(nk_n -1)}{n(n-1)k_n(k_n-1)}({\hat{\vtheta}}^{-\ell}-\vtheta^{-\ell})^T\sum_{t=1}^{n} \mathcal{B}(\vX_{t-1}^{-\ell})\sum_{s \in W_t}v_s\nonumber\\
&-& \sqrt{\frac{n}{k_n}}\frac{k_n}{n(n-1)}({\hat{\vtheta}}^{-\ell}-\vtheta^{-\ell})^T\sum_{t=1}^{n}v_t\sum_{s=1}^{n} \mathcal{B}(\vX_{s-1}^{-\ell})  \nonumber\\
&-&\sqrt{\frac{n}{k_n}} \frac{k_n}{n(k_n-1)}({\hat{\vtheta}}^{-\ell}-\vtheta^{-\ell})^T\sum_{t=1}^{n}v_t \mathcal{B}(\vX_{t-1}^{-\ell})\nonumber\\
&&+ \sqrt{\frac{n}{k_n}}\frac{(nk_n -1)}{n(n-1)k_n(k_n-1)}O_p(n^{-1}k_n^2)\sum_{t=1}^{n}\left|\sum_{s \in W_t}v_s\right|,
\eqan
where the last equality follows from condition C3 and Lemma 1.0.2 in \cite{ZambomAkritas2014}.
Because $v_t$ is white noise, it is easy to see that conditionally on $\vX$, $E(v_t\mathcal{B}(\vX_{t-1}^{-\ell})) = 0$ and 
\bqa
E\left\{\left[v_t\mathcal{B}(\vX_{t-1}^{-\ell})\right]^2\right\} = \tau\mathcal{B}^2(\vX_{t-1}^{-\ell}).
\eqa
Thus, since each of the $d$ components of $({\hat{\vtheta}}^{-i}-\vtheta^{-i})$ is $O_p(n^{-1/2})$, the third term on the right hand side of (\ref{v_m}) is of order $O_p(n^{-1/2}k_n^{-1/2}n^{-1/2}n^{1/2}) = o_p(1)$. Similarly, the second term on the right hand side of (\ref{v_m}) is of order $O_p(n^{-3/2}k_n^{1/2}n^{-1/2}n^{1/2}n) = O_p(n^{-1/2}k_n^{1/2}) = o_p(1)$. The first term on the right hand side of (\ref{v_m}) is of order $O_p(n^{-1/2}k_n^{-3/2}n^{-1/2}nk_n^{1/2}) = O_p(k_n^{-1}) = o_p(1)$. The last term on the right hand side of (\ref{v_m}) is of order $O_p(n^{-1/2}k_n^{-3/2}n^{-1}k_n^2nk_n) = O_p(n^{-1/2}k_n^{3/2}) = o_p(1)$ by the choice of $k_n$.

Next, re-write the third term in equation \eqref{T-stat}, i.e., $\vm_V'\vA\vm_V$, as
\bqan\label{m_m}
\sqrt{\frac{n}{k_n}}\vm_V'\vA\vm_V \hspace{-.2cm} &=& \hspace{-.2cm} \sqrt{\frac{n}{k_n}}\frac{(nk_n-1)}{n(n-1)k_n(k_n-1)}\sum_{t=1}^{n}\left(\sum_{s\in W_t}({\hat{\vtheta}}^{-\ell}-\vtheta^{-\ell})^T \mathcal{B}(\vX_{s-1}^{-\ell})\right)^2 \nonumber \\
&-&\sqrt{\frac{n}{k_n}}\frac{k_n}{n(n-1)}\left(\sum_{t=1}^{n}({\hat{\vtheta}}^{-\ell}-\vtheta^{-\ell})^T \mathcal{B}(\vX_{t-1}^{-\ell})\right)^2\nonumber\\
&-&\sqrt{\frac{n}{k_n}}\frac{k_n}{n(k_n-1)}\sum_{t=1}^{n}(({\hat{\vtheta}}^{-\ell}-\vtheta^{-\ell})^T \mathcal{B}(\vX_{t-1}^{-\ell}))^2.
\eqan
The first term on the right hand side of (\ref{m_m}) is of order $O_p(n^{-1/2}k_n^{-3/2}n(k_nn^{-1/2})^2) = o_p(1)$, the second term is of order $O_p(n^{-3/2}k_n^{1/2}(nn^{-1/2})^2) = o_p(1)$, while the last term is of order $O_p(n^{-1/2}k_n^{-1/2}nn^{-1/2}) = o_p(1)$.

Consider the last term in equation \eqref{T-stat}:
\bqan  \label{alpha_m}
\sqrt{\frac{n}{k_n}}\valpha_V'\vA\vm_V &=& \sqrt{\frac{n}{k_n}}\frac{(nk_n -1)}{n(n-1)k_n(k_n-1)}(\hat\alpha_1 - \alpha_1)\sum_{t=1}^{n}\sum_{s \in W_t} \epsilon^2_{s-1}\sum_{\ell\in W_t}({\hat{\vtheta}}^{-\ell}-\vtheta^{-\ell})^T \mathcal{B}(\vX_{s-1}^{-\ell}) \nonumber \\
 &-& \sqrt{\frac{n}{k_n}}\frac{k_n}{n(n-1)}(\hat\alpha_1 - \alpha_1)\sum_{t=1}^{n}\epsilon^2_{t-1}\sum_{\ell=1}^{n}({\hat{\vtheta}}^{-\ell}-\vtheta^{-\ell})^T \mathcal{B}(\vX_{s-1}^{-\ell}) \nonumber\\
 &-&\sqrt{\frac{n}{k_n}} \frac{k_n}{n(k_n-1)}(\hat\alpha_1 - \alpha_1)\sum_{t=1}^{n}\epsilon^2_{t-1}({\hat{\vtheta}}^{-\ell}-\vtheta^{-\ell})^T \mathcal{B}(\vX_{t-1}^{-\ell}).
\eqan
The first term is of order $O_p(n^{-1/2}k_n^{-3/2}n^{-1/2}nk_n^{1/2}k_nn^{-1/2}) = o_p(1)$, and the second term is of order $O_p(n^{-3/2}k_n^{1/2}n^{-1/2}n^{1/2}nn^{-1/2}) = o_p(1)$. Moreover, since $v_i$ is white noise, the third term is of order $O_p(n^{-1/2}k_n^{-1/2}n^{-1/2}n^{1/2}n^{-1/2}) = o_p(1)$.

\end{proof}

\subsection{Proof of Theorem \ref{thm.altern}}
\begin{proof} As in the proof of Theorem \ref{thm.null} we here prove the case of $p = 1$; the proof for $p > 1$ can be shown in a similar way. Under the additive alternative hypothesis in (\ref{Ha}) the residuals are

\bqa
\hat{v}_{t} &=&  \epsilon_t^2  - \hat{\alpha}_1\epsilon_{t-1}^2 - \hat{m}_\ell(\vX^{-\ell}_{t-1})\nonumber\\
&=&  \epsilon_t^2 - \alpha_1\epsilon_{t-1}^2 - m_\ell(\vX^{-\ell}_{t-1}) - (nk_n)^{-1/4}\tilde{m}_\ell(X_{it})  - (\hat{\alpha}_1 - \alpha_1)\epsilon_{t-1}^2 - (\hat{m}_\ell(\vX^{-\ell}_{t-1}) - m_\ell(\vX^{-\ell}_{t-1}))\\
&& + (nk_n)^{-1/4}\tilde{m}_\ell(X_{\ell t})\nonumber\\
&=&v_t -  (\hat{\alpha}_1 - \alpha_1)\epsilon_{t-1}^2 - ({\hat{\vtheta}}^{-\ell}-\vtheta^{-\ell})^T \mathcal{B}(\vX_{t-1}^{-\ell}) + (nk_n)^{-1/4}\tilde{m}_\ell(X_{\ell t}).
\eqa

Then, the statistic $T_n = MST_n - MSE_n$ can be written as
\bqan \label{VAV_alter}
T &=& \vV'\vA\vV = (\vv_V + \valpha_V + \vm_V)'\vA(\vv_V + \valpha_V + \vm_V) +  2(nk_n)^{-1/4}(\vv_V + \valpha_V + \vm_V)'\vA\tilde{\vm}_{iV}\nonumber\\&& +  ((nk_n)^{-1/4})^2\tilde{\vm}_{iV}'\vA\tilde{\vm}_{iV}\nonumber\\
&=& (\vv_V + \valpha_V + \vm_V)'\vA(\vv_V + \valpha_V + \vm_V)\nonumber\\
&& + 2(nk_n)^{-1/4}\vv_V'\vA\tilde{\vm}_{iV}
+ \valpha_V'\vA\tilde{\vm}_{iV}
+ \vm_V'\vA\tilde{\vm}_{iV} +  ((nk_n)^{-1/4})^2\tilde{\vm}_{iV}'\vA\tilde{\vm}_{iV}
\eqan
where $\tilde{\vm}_{iV}$ is defined as $\vv_V$ but with $\tilde{m}_i(X_{it})$ instead of $\hat{v}_{t}$.

By Theorem \ref{thm.null}, $n^{1/2}k_n^{-1/2}(\vv_V + \valpha_V + \vm_V)'\vA(\vv_V + \valpha_V + \vm_V)$ converges in distribution to a $N(0, 4\tau^2/3)$.

Consider now the second term on the right hand side of (\ref{VAV_alter}) and write
\bqan\label{v_m_tilde}
n^{1/2}k_n^{-1/2}(nk_n)^{-1/4}\vv_V'\vA\tilde{\vm}_{iV}
&=& 
 \frac{n^{1/4}}{k_n^{-3/4}} \frac{nk_n-1}{n(n-1)k_n(k_n-1)}\sum_{t=1}^n\sum_{k\in W_t}v_k\sum_{s\in W_t}\tilde{m}_\ell(X_{\ell s-1})\nonumber\\
&& 
+ \frac{n^{1/4}}{k_n^{-3/4}} \frac{k_n}{n(n-1)}\sum_{t=1}^nv_t\sum_{s=1}^n \tilde{m}_\ell(X_{\ell s-1}) + \frac{n^{1/4}}{k_n^{-3/4}} \frac{k_n}{n(k_n-1)}\sum_{t=1}^nv_t \tilde{m}_\ell(X_{\ell t-1})
\eqan

Using Lemma 1.0.2 in \citet{ZambomAkritas2014} and the Lipschitz continuity of $\tilde{m}_i(\cdot)$, the sum in the first term in (\ref{v_m_tilde}) is
\bqa
&&k_n\sum_{t=1}^n\left[\tilde{m}_\ell(X_{\ell t-1}) + O_p\left(\frac{k_n}{\sqrt{n}}\right)\right]\sum_{s\in W_t}v_s\\
&\leq&k_n\sum_{t=1}^n\left[\sum_{k\in W_t}\tilde{m}_\ell(X_{\ell k-1})\right]v_t + k_n^2O_p\left(\frac{k_n}{\sqrt{n}}\right)\sum_{t=1}^n|v_t| \\
&=&  k_n^2\sum_{t=1}^nv_t \tilde{m}_\ell(X_{t-1}) + O_p(k_n^3n^{1/2}).
\eqa
Furthermore, we find
\bqa
 &&n^{1/2}k_n^{-1/2}(nk_n)^{-1/4}\vv_V'\vA\tilde{\vm}_V\\
 &&  = \frac{n^{1/4}}{k_n^{-3/4}}\frac{nk_n}{n-1}\left[\left(\frac{1}{n}\sum_{t=1}^nv_t \tilde{m}_\ell(X_{\ell t-1})\right) - \left(\frac{1}{n}\sum_{t=1}^n\tilde{m}_\ell(X_{\ell t-1})\right)\left(\frac{1}{n}\sum_{t=1}^nv_t\right)\right]\\
 && 
 + \frac{n^{1/4}}{k_n^{-3/4}}O_p\biggl(\frac{k_n^2}
 {n^{1/2}}\biggr),
\eqa
which goes to 0 in probability, since $E(L) = 0$ and $Var(L) = O(n^{-1})$. Hence, $L = O_p(n^{-1/2})$, where $L = \left(n^{-1}\sum_{t=1}^nv_t \tilde{m}_\ell(X_{\ell t-1})\right) - \left(n^{-1}\sum_{t=1}^n\tilde{m}_\ell(X_{\ell t-1})\right)\left(n^{-1}\sum_{t=1}^nv_t\right)$.

The third and fourth terms on the right hand side of (\ref{VAV_alter}) can be shown to be $o_p(1)$ using similar steps. 

Now consider the last term in (\ref{VAV_alter}) and write
\bqan\label{m_tilde_m_tilde}
n^{1/2}k_n^{-1/2}(nk_n)^{-1/2}\tilde{\vm}_{V\ell}'\vA\tilde{\vm}_{V\ell}
&=& 
 \frac{n^{1/4}}{k_n^{-3/4}} \frac{nk_n-1}{n(n-1)k_n(k_n-1)}\sum_{t=1}^n\sum_{k\in W_t}\tilde{m}_\ell(X_{\ell k-1})\sum_{s\in W_t}\tilde{m}_\ell(X_{\ell s-1})\nonumber\\
&& 
+ \frac{n^{1/4}}{k_n^{-3/4}} \frac{k_n}{n(n-1)}\sum_{t=1}^n\tilde{m}_\ell(X_{\ell t-1})\sum_{s=1}^n \tilde{m}_\ell(X_{\ell s-1})\nonumber\\
&&+ \frac{n^{1/4}}{k_n^{-3/4}} \frac{k_n}{n(k_n-1)}\sum_{t=1}^n \tilde{m}^2_\ell(X_{\ell t-1})
\eqan
which can be written as
\bqa
\frac{n}{n-1}\left[\left(\frac{1}{n}\sum_{t=1}^n\tilde{m}_\ell^2(X_{\ell t-1})\right) - \left(\frac{1}{n}\sum_{t=1}^n\tilde{m}_\ell(X_{\ell t-1})\right)^2\right] 
+ O_p\biggl(\frac{k_n^2}{n^{1/2}}\biggr)
\pkonv Var(\tilde{m}_\ell(X_\ell)).
\eqa

For part b) of the Theorem note that the statistic $T_n = MST_n - MSE_n$ can be written as
\bqan \label{VAV_alterG}
T = \vV'\vA\vV &=& (\vv_V + \valpha_V + \vm_V + (nk_n)^{-1/4}\tilde{\vm}_{V\ell})'\vA(\vv_V + \valpha_V + \vm_V + (nk_n)^{-1/4}\tilde{\vm}_{V\ell})  \nonumber\\
&&+ 2(nk_n)^{-1/4}(\vv_V + \valpha_V + \vm_V + (nk_n)^{-1/4}\tilde{\vm}_{V\ell})'\vA\tilde{\vm}_V + (nk_n)^{-1/2}\tilde{\vm}_V'\vA\tilde{\vm}_V
\eqan
where $\tilde{\vm}_{V}$ is defined as $\vv_V$ but with $\tilde{m}(\vX_{t})$ instead of $\hat{v}_{t}$.
Using part a) and Theorem \ref{thm.null}, the first term on the right hand side of (\ref{VAV_alterG}) goes to $Var(\tilde{m}_i(X_i))$ in probability.

Consider the last term on the right hand side of (\ref{VAV_alterG}).
Using steps similar to those in part a) it is straightforward to show that 
\bqa
n^{1/2}k_n^{-1/2}(nk_n)^{-1/2}\tilde{\vm}_V'\vA\tilde{\vm}_V \pkonv Var(\tilde{m}(\vX_t)).
\eqa
Finally, consider the second term on the right hand side of (\ref{VAV_alterG}). Write
\bqan \label{v_alpha_m_m_i}
&& n^{1/2}k_n^{-1/2}2(nk_n)^{-1/4}(\vv_V + \valpha_V + \vm_V + (nk_n)^{-1/4}\tilde{\vm}_{V\ell})'\vA\tilde{\vm}_V \nonumber \\
&=& 
 \frac{n^{1/4}}{k_n^{3/4}} \frac{nk_n-1}{n(n-1)k_n(k_n-1)}\sum_{t=1}^n\sum_{k\in W_t}\left(v_k - (\hat{\alpha}_1 - \alpha_1)\epsilon^2_{k-1} - (\hat{m}_\ell(\vX^{-\ell}_{k-1}) - m_\ell(\vX^{-\ell}_{k-1})) - (nk_n)^{-1/4}\tilde{m}_\ell(X_{\ell k})\right)\times \nonumber \\
 && \hspace{4.8cm}
\times \sum_{\ell\in W_t}\tilde{m}(\vX_{\ell-1})\nonumber\\
&& 
+ \frac{n^{1/4}}{k_n^{3/4}} \frac{k_n}{n(n-1)}\sum_{t=1}^n\left(v_t - (\hat{\alpha}_1 - \alpha_1)\epsilon^2_{t-1} - (\hat{m}_\ell(\vX^{-\ell}_{t-1}) - m_\ell(\vX^{-\ell}_{t-1})) - (nk_n)^{-1/4}\tilde{m}_\ell(X_{\ell t})\right)\sum_{s=1}^n \tilde{m}(\vX_{\ell s-1})  \nonumber\\
&& + 
\frac{n^{1/4}}{k_n^{3/4}} \frac{k_n}{n(k_n-1)}\sum_{t=1}^n\left(v_t - (\hat{\alpha}_1 - \alpha_1)\epsilon^2_{t-1} - (\hat{m}_\ell(\vX^{-\ell}_{t-1}) - m_\ell(\vX^{-\ell}_{t-1})) - (nk_n)^{-1/4}\tilde{m}_\ell(X_{\ell t})\right) \tilde{m}(\vX_{t-1}).
\eqan
Note that $n^{-3/4}k_n^{1/4}\sum_{t=1}^n \left(v_t - (\hat{\alpha}_1 - \alpha_1)\epsilon^2_{t-1} - (\hat{m}_\ell(\vX^{-\ell}_{t-1}) - m_\ell(\vX^{-\ell}_{t-1}))\right)\tilde{m}(\vX_{t-1}) \pkonv 0$ and \\$n^{-1}k_n^{-3/4}\sum_{t=1}^n\tilde{m}_\ell(X_{\ell t})\tilde{m}(\vX_{t-1}) \pkonv 0$
so that the last term on the right hand side of (\ref{v_alpha_m_m_i}) goes in probability to 0.

Consider the second term on the right hand side of (\ref{v_alpha_m_m_i}). Note that $n^{-1}\sum_{t=1}^n\tilde{m}(\vX_{t-1}) \pkonv E(\tilde{m}(\vX))$ but $n^{-3/4}k_n^{1/4}\sum_{t=1}^n\left(v_t - (\hat{\alpha}_1 - \alpha_1)\epsilon^2_{t-1} - (\hat{m}_\ell(\vX^{-\ell}_{t-1}) - m_\ell(\vX^{-\ell}_{t-1})) - (nk_n)^{-1/4}\tilde{m}_\ell(X_{\ell t})\right) \pkonv 0$, so that the second term on the right hand side of (\ref{v_alpha_m_m_i}) goes to 0. The first term on the right hand side of (\ref{v_alpha_m_m_i}) goes to 0 in probability because $n^{-3/4}k_n^{1/4}\sum_{t=1}^nv_t\tilde{m}(\vX_{t-1})$, or any other term instead of $v_t$ goes to 0 in probability. This completes the proof.

\end{proof}

\bibliographystyle{apalike}
\bibliography{mybibfile}

\end{document}